\documentclass[runningheads]{llncs}
 
\usepackage{pdfpages}
\usepackage{stmaryrd}
\usepackage{comment}
\usepackage[pdfencoding=auto]{hyperref}
\usepackage{times}

\usepackage[ruled,vlined,linesnumbered]{algorithm2e}

\usepackage{graphicx,booktabs,listings,caption,subcaption,tikz,amsmath,amssymb,placeins,subfiles,tabto}
\usepackage{xcolor,colortbl}
\usepackage{multirow,float}

\usepackage{mathtools,centernot}
\usepackage{tikz, trimclip}

\usepackage{cleveref}

\usepackage{pifont}

\usepackage{lineno}

\newcommand{\tent}[1]{\textcolor{blue}{#1}}

\definecolor{azure}{rgb}{0.0, 0.5, 1.0}
\definecolor{cadmiumorange}{rgb}{0.93, 0.53, 0.18}
\definecolor{cadmiumgreen}{rgb}{0.0, 0.42, 0.24}
\definecolor{cadmiumred}{rgb}{0.89, 0.0, 0.13}
\definecolor{cadmiumyellow}{rgb}{1.0, 0.96, 0.0}
\definecolor{applegreen}{rgb}{0.55, 0.71, 0.0}

\newcommand{\co}[1]{\textcolor{cadmiumorange}{\underline{#1\mathstrut}}}

\newcolumntype{L}{>{$}c<{$}}
\newcommand{\mode}              {\mu}
\newcommand{\F}                 {\mathcal{F}}
\newcommand{\M}                 {\mathcal{M}}

\newcommand{\R}                 {\mathcal{R}}
\newcommand{\Nat}               {\mathbb{N}}
\newcommand{\Int}               {\mathbb{Z}}
\newcommand{\supp}              {\mathit{supp}}
\newcommand{\barp}              {\bar{p}}
\newcommand{\nunet}             {$\nu$-net}
\newcommand{\nunets}             {$\nu$-nets}
\newcommand{\nuN}               {\nu\textit{-}\mathit{N}}

\newcommand{\cutmarking}        {\widetilde{m}}
\newcommand{\compal}            {{\widetilde{\gamma}}}

\newcommand{\matr}[1]           {\mathbf{#1}}

\newcommand{\proj}              {{\upharpoonright}}

\DeclareMathOperator{\viol}{viol}
\DeclareMathOperator{\alignm}{align}

\DeclareMathOperator{\act}{activity}

\DeclareMathOperator{\ecase}{case}
\DeclareMathOperator{\Res}{Res}

\newcommand{\clm}             {\downarrow}
\newcommand{\rls}             {\uparrow}

\newcommand{\set}[1]{\{ #1 \}}
\newcommand{\sset}[1]{\left\{ #1 \right\}}
\newcommand{\mset}[1]{[ #1 ]}

\newcommand{\seq}[1]{\langle #1 \rangle}

\newcommand{\spair}[2]{\left(#1, #2 \right)}

\newcommand{\interval}[2]          {[#1, #2 ] }
\newcommand{\intervalopen}[2]      {(#1, #2)}
\newcommand{\intervalhalfopenl}[2] {(#1, #2]}
\newcommand{\intervalhalfopenr}[2] {[#1, #2)}
\newcommand{\prefix}[1]            {\intervalhalfopenl{\bot}{#1}}
\newcommand{\prefixopen}[1]        {\intervalopen{\bot}{#1}}
\newcommand{\postfix}[1]           {\intervalhalfopenr{#1}{\top}}
\newcommand{\postfixopen}[1]       {\intervalopen{#1}{\top}}

\newcommand{\vect}[1]     {\mathbf{#1}}

\newcommand{\pre}[1]             {\prescript{\bullet}{}{#1}}
\newcommand{\post}[1]            {#1^\bullet}

\newcommand{\idc}               {c}
\newcommand{\idr}               {\rho}
\newcommand{\Idc}               {\textnormal{Id}_c}
\newcommand{\Idr}               {\textnormal{Id}_r}
\newcommand{\IdR}               {\textnormal{Id}_R}

\newcommand{\e}                 {\varepsilon}

\newcommand{\ie}{i.e., }

\newcommand{\eg}{e.g., }

\newcommand{\Var}               {\mathit{Var}}
\newcommand{\Id}                {\mathit{Id}}

\newcommand{\arxiv}{1}

\begin{document}

\title{
Exact and Approximated Log Alignments for Processes with Inter-case Dependencies
}

\titlerunning{Exact and Approximated Log Alignments for Processes with Inter-case Dependencies}

\author{Dominique Sommers
\and
Natalia Sidorova
\and
Boudewijn van Dongen
}
\authorrunning{D. Sommers et al.}

\institute{Department of Mathematics and Computer Science\\Eindhoven University of Technology, Eindhoven, the Netherlands\\
\email{\{\href{mailto:d.sommers@tue.nl}{d.sommers},\href{mailto:n.sidorova@tue.nl}{n.sidorova},\href{mailto:b.f.v.dongen@tue.nl}{b.f.v.dongen}\}@tue.nl}}

\maketitle 

\begin{abstract}
The execution of different cases of a process is often restricted by inter-case dependencies through \eg queueing or shared resources. Various high-level Petri net formalisms have been proposed that are able to model and analyze coevolving cases. In this paper, we focus on a formalism tailored to conformance checking through alignments, which introduces challenges related to constraints the model should put on interacting process instances and on resource instances and their roles. We formulate requirements for modeling and analyzing resource-constrained processes, compare several Petri net extensions that allow for incorporating inter-case constraints. We argue that the Resource Constrained \nunet~is an appropriate formalism to be used the context of conformance checking, which traditionally aligns cases individually failing to expose deviations on inter-case dependencies. We provide formal mathematical foundations of the globally aligned event log based on theory of partially ordered sets and propose an approximation technique based on the composition of individually aligned cases that resolves inter-case violations locally.

\end{abstract}

\section{Introduction}\label{sec:introduction}
Event logs record which activity is executed at which moment of time, and additionally they often include indications which resources were involved in which activity, mentioning the exact person(s) or machine(s). The availability of such event logs enables the use of conformance checking for resource-constrained processes, analyzing not only the single instance control-flow perspective, but also checking whether and where the actual process behavior recorded in an event log deviates from the resource constraints prescribed by a process model.

Process models, and specifically Petri nets with their precise semantics, are often used to describe and reason about the execution of a process. In many approaches, a process model considers a process instance (a case) in isolation from other cases~\cite{van1998application}. 
In practice, however, a process instance is usually subject to interaction with other cases and/or resources, whose availability puts additional constraints on the process execution. 
In order to expose workflow deviations caused by inter-case dependencies, it is crucial to use models considering multiple cases simultaneously. 

There are several approaches to modeling and analysis of processes with inter-case dependencies. 
In~\cite{barkaoui1998structural} and~\cite{van2006resource}, Petri nets are extended with resources to model availability of durable resources, with multiple cases competing by claiming and releasing these \emph{shared} resources. 
To distinguish the cases, $\nu$-Petri nets~\cite{rosa2010decision} incorporate name creation and management as a minimal extension to classical Petri nets, with the advantage that coverability and termination are still decidable, opposed to more advanced Petri net extensions.
The functionality of $\nu$-Petri nets is inherited in other extensions such as Catalog Petri nets~\cite{ghilardi2020petri}, synchronizing proclet models~\cite{fahland2019describing}, resource and instance-aware workflow nets (RIAW-nets)~\cite{montali2016model}, DB-nets~\cite{montali2017db} and resource constrained $\nu$-Petri nets~\cite{sommers2022aligning}, all with the ability to handle multiple cases simultaneously. For the latter, the cases are assumed to follow the same process, interacting via (abstract) shared resources in a one-to-many relation, \ie a resource instance can be claimed by one case at a time. More sophisticated extensions allow for cases from various perspectives with many-to-many interactions, via \eg concepts from databases, shared resources and proclet channels. This may impose, however, problems of undecidability during conformance checking, which we discuss in this work.

Many conformance checking techniques use \emph{alignments} to expose where the behavior recorded in a log and the model agree, which activities prescribed by the model are missing in the log and which log activities should not be performed according to the model~\cite{carmona2018conformance,van2012replaying}. The usual focus is on the control flow of the process. In more advanced techniques~\cite{alizadeh2018linking,de2013aligning,mannhardt2016balanced,mehr2021detecting}, data and/or resource information is additionally incorporated in the alignments by considering these perspective only after the control flow~\cite{de2013aligning}, by balancing the different perspectives in a customizable manner~\cite{mannhardt2016balanced} or by considering all perspectives at once~\cite{mehr2021detecting}. These three types of techniques operate on a case-by-case basis, which can lead to misleading results in case of shared resources, \eg when multiple cases claim the same resource simultaneously.

In our previous work we considered the execution of all process instances by aligning the complete event log to a resource constrained $\nu$-Petri nets~\cite{sommers2022aligning}. In this paper, we present our further steps: (1) We compare how the existing Petri net extensions support modeling and analysis of processes with inter-case dependencies by formulating the requirements to such models, and we argue that \nunets~are an appropriate formalism. (2) We employ the poset theory to provide mathematical foundations for aligning the complete event log and exposing deviations of inter-case dependencies; (3) We propose an approximation method for computing optimal alignments in practice, which tackles the limitation of the computational efficiency when computing the complete event log alignment. The approximation method is based on composing alignments for isolated cases first and then resolving inter-case conflicts and deviations in the log locally.

The paper is organized as follows. In Section~\ref{sec:preliminaries} we introduce basic concepts of the poset theory, Petri nets and event logs. In Section~\ref{sec:resource_nets_modeling} we compare different Petri net extensions. We provide the mathematical foundations of the complete event log alignment in Section~\ref{sec:log_alignments}. Section~\ref{sec:approximation} presents the approximation method for computing alignments. We discuss implications of our work in Section~\ref{sec:conclusion}.

\section{Preliminaries}\label{sec:preliminaries}
In this section, we introduce basic concepts related to Petri nets and event logs and present the notations that we will use throughout the paper.

\subsection{Multisets and posets}
We start with definitions and notation regarding multisets and partially ordered sets.

\begin{definition}{(Multiset)}
A \emph{multiset} $m$ over a set $X$ is $m: X \rightarrow \Nat$. $X^\oplus$ denotes the set of all multisets over $X$. We define the support $\supp(m)$ of a multiset $m$ as the set $\{x \in X \mid m(x) > 0\}$. We list elements of the multiset as $[m(x) \cdot x \mid x \in X]$, and write $|x|$ for $m(x)$, when it is clear from context which multiset it concerns.

For two multisets $m_1,m_2$ over $X$, we write $m_1 \leq m_2$ if $\forall_{x \in X} m_1(x) \leq m_2(x)$, and $m_1 < m_2$ if $m_1 \leq m_2 \wedge m_1 \neq m_2$. We define $m_1 + m_2 = [(m_1(x) + m_2(x))\cdot x \mid x \in X]$, and $m_1 - m_2 = [\max(0, m_1(x) - m_2(x))\cdot x \mid x \in X]$ for $m_1 \geq m_2$. 

Furthermore, $m_1 \sqcup~ m_2 = [\max(m_1(x), m_2(x))\cdot x \mid x \in X]$, $m_1 \sqcap m_2 = [\min(m_1(x), m_2(x))\cdot x \mid x \in X]$.
\end{definition}
In some cases, we consider multisets over a set $X$ as vectors of length $|X|$, assuming an arbitrary but fixed ordering of elements of $X$.

\begin{definition}{(Partial order, Partially ordered set, Antichains)}
A \emph{partially ordered set} (\emph{poset}) $X = (\bar X, \prec_X)$ is a pair of a set $\bar X$ and a \emph{partial order} $\prec_X \subseteq X \times X$. 
We overload the notation and write $x \in X$ if $x \in \bar X$. For $x,y \in X$, we write $x \|_X y$ if $x \nprec y \wedge y \nprec x$ and $x \preceq y$ if $x \prec y \vee x = y$.

Given $\prec_X$, we define $\prec_X^+$ to be the smallest transitively closed relation containing $\prec_X$. Thus $\prec_X^+$ is a partial order with $\prec_X \subseteq \prec_X^+$. 

We extend the standard set operations of union, intersection, difference and subsets to posets: for any two posets $X$ and $Y$, $X \circ Y = (\bar{X} \circ \bar{Y}, (\prec_X \circ \prec_Y)^+)$, with $\circ \in \set{\cup, \cap, \setminus}$ and $Y \subseteq X$ iff $\bar{Y} \subseteq \bar{X}$ and $\prec_Y = \prec_X \cap (\bar{Y} \times \bar{Y})$.

A poset $A$ is an \emph{antichain} if no elements of $A$ are comparable, \ie  $\forall_{x,y \in A}~x \| y$. For poset $X$, $\mathcal{A}(X)$ denotes the set of all \emph{antichains} $A\subseteq X$, and  $\mathcal{A}^+(X)$ is the set of all \emph{maximal antichains}: $\mathcal{A}^+(X)=\set{A \mid A \in \mathcal{A}(X), \forall_{B \in \mathcal{A}(X)}~B \subseteq A \implies B = A}$.

Two special maximal antichains are the \emph{minimum} and \emph{maximum} elements of $X$, defined by $\min(X) = \set{x \mid x \in X, \forall_{y \in X} y \nprec x} \in \mathcal{A}^+(X)$ and $\max(X) = \set{x \mid x \in X, \forall_{y \in X} x \nprec y} \in \mathcal{A}^+(X)$.

We define $X^< = \set{ (\bar{Y}, \prec_Y) \mid \bar{Y} = \bar{X}, \prec_X \subseteq \prec_Y, \forall_{a,b \in Y, a \neq b}~a \not\|_Y b }$ to be the set of totally ordered permutations of $X$ that respect the partial order.
\end{definition}

\begin{definition}{(Interval, prefix and postfix in a poset)}\label{def:poset_interval}
With a poset $X$ and two antichains $A,B \in \mathcal{A}(X)$, the \emph{closed jhkcbvinterval} from $A$ to $B$ is the subposet defined as follows: $\interval{A}{B} = (\overline{AB}, \prec_X \cap (\overline{AB}\times\overline{AB}))$ with $\overline{AB} = \set{x \mid x \in X, A \preceq x \preceq B}$, 
and the half open and open intervals:  $\intervalhalfopenl{A}{B} = \interval{A}{B}\setminus A$,  $\intervalhalfopenr{A}{B} =  \interval{A}{B}\setminus B$ and $\intervalopen{A}{B} = \intervalhalfopenr{A}{B} \setminus A$.

Artificial minimal and maximal elements are denoted as $\bot$ and $\top$ respectively, \ie $\forall_{x \in X} \bot \prec x \prec \top$. $\prefix{A}$, $\prefixopen{A}$, $\postfix{A}$ $\postfixopen{A}$ denote the corresponding \emph{prefixes and postfixes} of an antichain $A \in \mathcal{A}(X)$ in $X$.
\end{definition}

\subsection{Petri nets}
Petri nets can be used as a tool for the representation, validation and verification of workflow processes to provide insights in how a process behaves~\cite{peterson1981petri}.

\begin{definition}{(Labeled Petri nets, Pre-set, Post-set)}
A \emph{labeled Petri net}~\cite{murata1989petri} is a tuple $N = (P, T, \F, \ell)$, with sets of places and transitions $P$ and $T$, respectively, such that $P \cap T = \emptyset$, and a multiset of arcs $\F: (P \times T) \cup (T \times P) \rightarrow \Nat$ defining the flow of the net. $\ell: T \rightarrow \Sigma^\tau = \Sigma \cup \{\tau\}$ is a \emph{labeling} function, assigning each transition $t$ a label $\ell(t)$ from alphabet $\Sigma$ or $\ell(u)=\tau$ for silent transitions.

We assume that the \emph{intersection, union and subsets} are only defined for two labeled Petri nets $N_1$, $N_2$ where $\forall_{t \in T_1 \cap T_2} \ell_1(t) = \ell_2(t)$.



Given an element $x \in P \cup T$, its \emph{pre- and post-set} $\pre{x}$ ($\post{x}$) are multisets defined by $\pre{x} = [\F(y, x) \cdot y \mid y \in P \cup T]$ and $\post{x} = [\F(x, y) \cdot y \mid y \in P \cup T]$ resp.
\end{definition}

\begin{definition}{(Marking, Enabling and firing of transitions, Reachable markings)}\label{def:marking}
A \emph{marking} $m \in P^\oplus$ of a (labeled) Petri net $N = (P, T, \F, \ell)$ assigns how many tokens each place contains and defines the state of $N$.

With $m$ and $N$, a transition $t \in T$ is \emph{enabled} for firing iff $m \geq \pre{t}$. We denote the \emph{firing} of $t$ by $m \xrightarrow{t} m'$, where $m'$ is the resulting marking after firing $t$ and is defined by $m' = m - \pre{t} + \post{t}$. For a transition sequence $\sigma = \seq{t_1,\dots,t_n}$ we write $m \xrightarrow{\sigma} m'$ to denote the consecutive firing of $t_1$ to $t_n$. We say that $m'$ is reachable from $m$ and write $m \xrightarrow{*} m'$ if there is some $\sigma \in T^*$ such that $m \xrightarrow{\sigma} m'$.

$\M(N)=P^\oplus$ and it denotes the set of all markings in net $N$ and $\R(N,m)$ the \emph{set of  markings reachable} in net $N$ from marking $m$.
\end{definition}

\begin{definition}{(Place invariant)}
Let $N = (P, T, \F, \ell)$ be a Petri net. A \emph{place invariant} \cite{lautenbach1975liveness} is a row vector $I: \bf{P} \rightarrow \mathbb{Q}$ such that $I \cdot \bf{F} = 0$, with $\bf{P}$ and $\bf{F}$ vector representations of $P$ and $\F$.  We denote the set of all place invariants as $\mathcal{I}_N$, which is a linear subspace of $\mathbb{Q}^P$. 
\end{definition}

The main property of a place invariant $I$ in a net $N$ with initial marking $m_i$ is that $\forall_{m_1, m_2 \in \R(N,m_i)} I \cdot m_1 = I \cdot m_2$.

\begin{definition}{(Net system, Execution poset and sequence, Language)}
A \emph{net system} is a tuple $SN = (N, m_i, m_f)$, where $N$ is a (labeled) Petri net, and $m_i$ and $m_f$ are respectively the initial and final marking. An \emph{execution sequence} in a net system $SN = (N, m_i, m_f)$ is a firing sequence from $m_i$ to $m_f$. Additionally, an \emph{execution poset} is a poset of transition firings, where each totally ordered permutation is a firing sequence. The \emph{language} of a net system $SN$ is the set of all execution sequences in $SN$.
\end{definition}

\subsection{Event logs}
An event log records activity executions as events including at least the occurred activity, the time of occurrence and the case identifier of the corresponding case.
Often resources are also recorded as event attributes, \eg the actors executing the action. It is generally known beforehand in which activities specific \emph{resource roles} $R$ are involved and which resource instances $\Idr$ are involved in the process for each role $r \in R$. We assume that each resource has only one role (function) allowing to execute a predefined number of tasks, and therefore define the set  $\IdR$ of resource instances of all roles as the disjoint union of resource instance sets of roles: $\IdR = \uplus_{r \in R} \Idr$. A \emph{resource instance} $\idr \in \IdR$ with role $r \in R$ is equipped with capacity, making $\Idr$ and $\IdR$ both multisets.

\begin{definition}{(Event, Event log, Trace)}\label{def:event_log}
An \emph{event} $e$ is a tuple $(a, t, c, \IdR')$, with an \emph{activity} $a = \act(e) \in \Sigma$, a \emph{timestamp} $t = time(e) \in \mathbb{R}$, a \emph{case identifier} $c = \ecase(e) \in \Idc$ and a multiset of resource instances $\IdR' = \Res(e) \leq \IdR$. Such an event represents that activity $a$ occurred at timestamp $t$ for case $c$ and is executed by resource instances from $\IdR'$ belonging to possibly different resource roles.

An \emph{event log} $L$ is a set of events with partial order $\prec_L$ that respects the chronological order of the events, \ie $\forall_{e_1, e_2 \in L} time(e_1) < time(e_2) \Longrightarrow e_2 \not\prec_L e_1$.
An event log can be partitioned into \emph{traces}, defined as projections \eg on the case identifiers or on the resources names. For every $c \in \Idc$, $L_c$ denotes a trace projected on the case identifier $c$ defined by $L_c = (\set{e \mid e \in L, \ecase(e) = c}, \prec_{L_c})$ with $\prec_{L_c} = \set{ (e, e') \mid (e, e') \in \prec_L, \ecase(e) = \ecase(e') = c }$.
\end{definition}

Alternatively, we write $\seq{e_1, e_2, \cdots}$ for an event log which is totally ordered, and $a^{\IdR^{'}}$ and $\co{a}^{\IdR^{'}}$ for events where the case is identified by the activity color (and bar position) and the time of occurrence is abstracted away from.

For a (labeled) Petri net modeling a process, the transitions' names or labels correspond to the activity names found in the recorded event log.

\section{Modeling, analysis and simulation of case handling systems with inter-case dependencies}\label{sec:resource_nets_modeling}
A classical Petri net models a process execution using transition firings and the corresponding changes of markings without making distinctions between different cases on which the modeled system works simultaneously. To create a case view, Workflow nets~\cite{van2016data} model processes from the perspective of a single case. Systems in which cases interact with each other, \eg by queueing or sharing resources, need to be modeled in a different way. We show from a modeling point how this boils down to multiple cases competing over shared tokens representing resources in a Petri net, which requires an extension on the formalism of the classical Petri nets. In Sec.~\ref{sec:req}, we motivate the requirements by providing examples, after which, in Sec.~\ref{sec:existing_extensions}, we discuss whether existing Petri net extensions satisfy these requirements. We end, in Sec.~\ref{sec:rc_nunet} by proposing a \emph{minimal extension} based on $\nu$-Petri nets~\cite{rosa2010decision} that meets each requirement for simulation and analysis of resource-constrained processes.

\subsection{Requirements imposed by inter-case dependencies}\label{sec:req}
When modeling systems with inter-case dependencies, \ie shared resources, simultaneous cases can interfere in each other's processing via the resources, causing inter-case dependencies. To model, simulate and analyze such behavior, the cases and resources, represented as tokens in a Petri net, should be handled together and simultaneously in the process model. This introduces the need for case (R1) and resource isolation (R2) as well as durable resources (R3) and case-resource correlations (R4), which regular Petri nets are not capable of. For analysis, like computing alignments (see Section~\ref{sec:log_alignments}), non-invertible functions can cause state-space explosions (R5). We show for each requirement, when not satisfied, how simulation and/or analysis concerning multiple simultaneous cases fails:
\begin{figure}[t]
\centering
\includegraphics[scale=0.5]{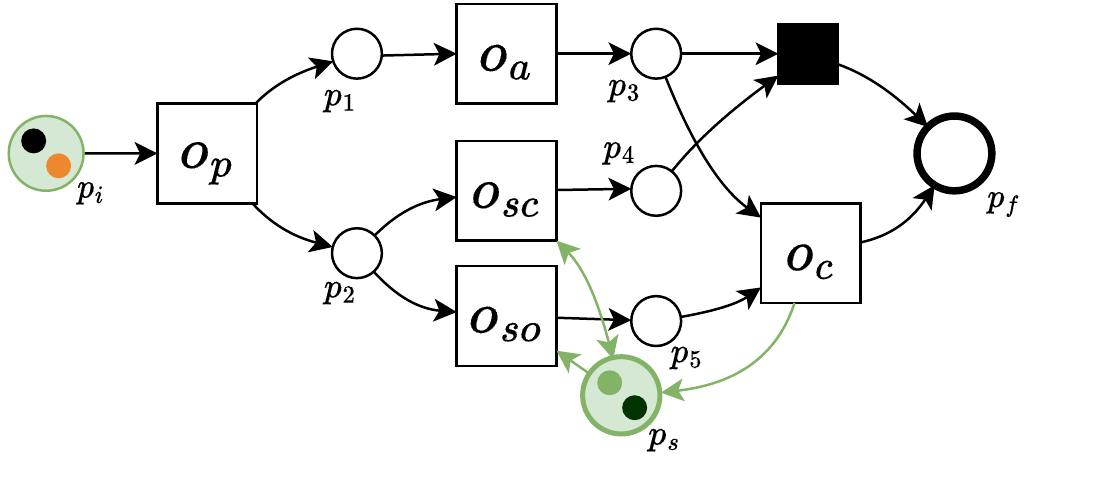}
    \caption{Example Petri net $N_1$ to argue the requirements, with token colors denoting different instances.}
    \label{fig:requirements}
\end{figure}
\begin{itemize}
    \item[R1] \emph{Distinguishable cases} are required when dealing with multiple cases. Tokens involved in a firing of a transition should not belong to different cases, unless case batching is used. Mixing tokens from different cases, possible in classical Petri nets, can potentially cause model behavior that is not possible in the modeled system: Suppose we have a simple operation process modeled by Petri net $N_1$, shown in Fig.~\ref{fig:requirements}, where a patient undergoes an operation involving the activities of preparation $(o_p)$, assistance $(o_a)$, closed surgery $(o_{sc})$ and open surgery $(o_{so})$ which is followed by closeup $(o_c)$. We assume case tokens to be indistinguishable. 
    The language of $(N_1,[p_i, 2p_s],[p_f, 2p_s])$ is $\{\seq{o_p, o_a, o_{sc}}, \seq{o_p, o_{sc}, o_a}, \seq{o_p, o_a, o_{so}, o_c},$\\$\seq{o_p, o_{so}, o_a, o_c}\}$ 
    and the language of the same net processing two cases with sufficient resources has to consist of all possible interleaving of two traces belonging to single cases. However, $\set{\seq{ o_p, o_a, o_{sc}, \co{o_p}, \co{o_a}, \co{o_{so}}, \mathbf{o_c} }}$ is included in the language of $(N_1,[2p_i,2p_s],[2p_f,2p_s])$, which is impossible to obtain by an interleaving of two single cases, as $o_c$ is never enabled after $o_{sc}$ fires. Here and later we use underlined symbols when referring to the second case in examples.
    From now on, we assume case tokens are distinguishable and we have $m_i(p_i) = (c, \co{c})$;
    
    \item[R2] \emph{Distinguishable resources} are required when resource instances are uniquely identifiable. If the tokens in $p_s$ are indistinguishable,   $\seq{\dots, o_{so}^{\set{x}}, \co{o_{sc}}^{\set{x}}, o_{c}^{\set{x}}}$ belongs to the language of $(N_1,[2p_i,2p_s],[2p_f,2p_s])$. However, resource instance $x$ can only be claimed by the second case after it has been released by the first case (by firing transition $o_c$), hence it should not be included in the language. From now on, we assume resource tokens are distinguishable and we have $m_i(p_s) = (x, y)$;
    \item[R3] Resources are required to be \emph{durable} when having a variable number of cases in the system simultaneously. In $N_1$, the resource instances in $p_s$ are modeled to be durable, since these instances are always released after being claimed. However, were arc $(o_c, p_s)$ to be removed, problems arise when observed behavior concerns more than two cases, since after transition $o_{so}$ fired twice, it is never enabled again, causing a deadlock;
    \item[R4] Capturing \emph{case-resource correlation} is required when dealing with multiple distinguishable cases and resources in order to keep track of which resource handles which case. Without it, the language of $(N_1,[2p_i,2p_s],[2p_f,2p_s])$ includes \eg $\seq{\dots, o_{so}^{\set{x}}, \co{o_{so}}^{\set{y}}, \co{o_c}^{\set{x}}, o_c^{\set{y}}}$, which is undesirable as resources $x$ and $y$ have switched cases after transition $o_{so}$ is fired twice. Case-resource correlation should ensure, in this case, that transition $o_c$ can only be fired using the same resource as was claimed by firing transition $o_{so}$;
    \item[R5] Operations on token values (\eg guards, arc inscriptions) should be \emph{invertible} and computable when aligning observed and modeled behavior in order to keep the problem decidable. Consider \eg that patients enter the process by their name and birthdate $v$, which is transformed to an identifier $c$ in the first transition by an operation $f(v)$ on $(o_p)$. When activity $o_p$ is missing for a patient, it is undecidable which value $v$ should be inserted for the firing of $o_p$ when $f$ is not invertible.
\end{itemize}

\subsection{Existing Petri net extensions}\label{sec:existing_extensions}
Several extensions on Petri nets have been proposed focusing on multi-case and/or multi-resource processes able to handle (some) inter-case dependencies. We go over each extension, describing how they satisfy (and violate) requirements listed in Sec.~\ref{sec:req}. We propose an extension, which combines concepts of the described extensions and satisfies all requirements.

\emph{Resource constrained workflow nets} (RCWF-nets)~\cite{van2006resource} are Petri nets extended with resource constraints, where resources are durable units: they are claimed and then released again (R3). They define structural criteria for its correctness.
\begin{definition}{(Resource-constrained workflow net~\cite{van2006resource})} \label{def:RCWF}
Let $R$ be a set of \emph{resource roles}. A net system $N = (P_p \uplus P_r, T, \F_p \uplus \F_r, m_i, m_f)$ is a \emph{resource-constrained workflow net (RCWF-net)} with the set $P_p$ of \emph{production places} and the set $P_r = \set{p_r \mid r \in R}$ of \emph{resource places} iff 
\begin{itemize}
    \item $\F_p: (P_p \times T) \cup (T \times P_p) \rightarrow \mathbb{N}$ and $\F_r: (P_r \times T) \cup (T \times P_r) \rightarrow \mathbb{N}$;
    \item $N_p = (P_p, T, \F_p, [m_i(p) \cdot p \mid p \in P_p], [m_f(p) \cdot p \mid p \in P_p])$ is a net system, called the production net of $N$.
\end{itemize}
\end{definition}
The semantics of Petri nets is extended by having colored tokens on production places (R1) and as resources are shared across all cases, tokens on resource places are colorless ($\neg$R2, $\neg$R4). A transition is enabled if and only if there are sufficient tokens on its incoming places using tokens of the same color on production places.

\emph{$\nu$-Petri nets}~\cite{rosa2010decision} are an extension of Petri nets with pure name creation and name management, strictly surpassing the expressive power of regular Petri nets and they essentially correspond to the minimal object-oriented Petri nets of~\cite{kummer2000undecidability}. In a $\nu$-Petri net, names can be created, communicated and matched which can be used to deal with authentication issues~\cite{rosa2007expressiveness}, correlation or instance isolation~\cite{decker2008instance}. Name management is formalized by replacing ordinary tokens by distinguishable ones, thus adding color the the Petri net. 
\begin{definition}{(\texorpdfstring{$\nu$}{nu}-Petri net~\cite{rosa2010decision})}\label{def:nu-Petri_net}
Let $\Var$ be a fixed set of variables. A \emph{$\nu$-Petri net} is a tuple $\nuN = \langle P,T,\F \rangle$, with a set of places $P$, a set of transitions $T$ with $P \cap T = \emptyset$, and a flow function
$\F:(P \times T) \cup (T \times P) \rightarrow \Var^\oplus$ such that $\forall_{t \in T}$, $\Upsilon \cap \pre{t} = \emptyset \;\wedge\; \post{t} \setminus \Upsilon \subseteq \pre{t}$, where $\pre{t} = \bigcup\limits_{p \in P} supp(\F(p,t))$ and $\post{t} = \bigcup\limits_{p \in P} supp(\F(t,p))$. $\Upsilon \subset \Var$ denotes a set of special variables ranged by $\nu, \nu_1, \dots$ to instantiate fresh names.

A marking of $\nuN$ is a function $m: P \rightarrow \Id^\oplus$. $\Id(m)$ denotes the set of names in $m$, i.e. $\Id(m) = \bigcup\limits_{p \in P} supp(m(p))$.

A mode $\mode$ of a transition $t$ is an injection $\mode: \Var(t) \rightarrow \Id$, that instantiates each variable to an identifier.

For a firing of transition $t$ with mode $\mode$, we write $m \xrightarrow{t_\mode} m'$. $t$ is enabled with mode $\mode$ if $\mode(\F(p,t)) \subseteq m(P)$ for all $p \in P$ and $\mode(\nu) \notin \Id(m)$ for all $\nu \in \Upsilon \cap \Var(t) = supp(\cup_{p \in P} \F(p,t))$. The reached state after the firing of $t$ with mode $\mode$ is the marking $m'$, given by:
\begin{equation}
m'(p) = m(p) - \mode(\F(p,t)) + \mode(\F(t,p))\text{ for all }p \in P
\end{equation}
We denote $T_\mode$ to be the set of all possible transition firings.
\end{definition}
$\nu$-Petri nets support instance isolation for cases and resources requiring the tokens involved in a transition firing to have matching colors (R1, R2). Due to the tokens having singular identifiers, correlation between cases and resources can not be captured ($\neg$R4).

\emph{Resource and instance-aware workflow nets} (RIAW-nets)~\cite{montali2016model}, are Petri nets combining the notions from above by defining similar structural criteria for handling resource constraints on top of $\nu$-Petri nets. However, the resource places are assumed to only carry black tokens, not allowing for resource isolation and properly capturing the case-resource correlation.

\emph{Synchronizing proclets}~\cite{fahland2019describing} are a type of Petri net that describe the behavior of processes with many-to-many interactions: unbounded dynamic synchronization of transitions, cardinality constraints limiting the size of the synchronization, and history-based correlation of token identities (R1,R2). This correlation is captured by message-based interaction, specifying attributes of a message as correlation attributes (R4). The correlation constraints are $C_{init}$, $C^{\subseteq}_{match}$ and $C^=_{match}$, for initializing the attributes, partially and fully matching them. $\nu$-Petri nets are at the basis of proclets handling multiple objects by separating their respective subnets. While the proclet formalism is sufficient for satisfying all requirements listed above, they extend to many-to-many relations, which lifts the restriction that a resource can only be claimed by a single case.

\emph{Object-centric Petri nets}~\cite{van2020discovering}, similarly to synchronizing proclets, describe the behavior of processes with multiple perspectives and one-to-many and many-to-many relations between the different object types. These nets are a restricted variant of colored Petri nets where places are typed, tokens are identifiable referring to objects (R1,R2), and transitions can consume and produce a variable number of tokens. Correlation can be achieved with additional places of combined types (R4). Again, due to many-to-many relations, our one-to-many restriction on resources is lifted.

\emph{Database Petri nets} (DB-nets)~\cite{montali2017db} are extensions of $\nu$-Petri nets with multi-colored tokens that allows for multiple types of objects and their correlation (R1,R2,R4). Additionally, they support underlying read-write persistent storage consisting of a relational database with full-fledged constraints. Special ``view'' places in the net are used to inspect the content of the underlying data, while transitions are equipped with database update operations. These are in the general sense not invertible causing undecidability ($\neg$R5).

\emph{Catalog Petri nets} (CLog-nets)~\cite{ghilardi2020petri} are similar to DB-nets, but without the ``write'' operations (R1,R2,R4). The queries from view places in DB-nets have been relocated to transition guards, relying solely on the ``read-only'' modality for a persistent storage, however suffering from the same undecidability problem as these guards are not invertible in the general sense ($\neg$R5).

\subsection{Resource constrained $\nu$-Petri net with fixed color types}\label{sec:rc_nunet}
We combine conceptual ideas from the extensions described above, by extending RIAW-nets, which inherit the modeling restrictions from RCWF-nets and name management from $\nu$-Petri nets, using concepts from DB-nets and CLog-nets.

The resource places from RCWF-nets model the availability of resource instances by tokens, which is insufficient to capture correlation of cases by which they are claimed and released. We propose a minimal extension \emph{resource constrained $\nu$-Petri nets} (RC \nunet) which additionally contain busy places $\bar{P}_r = \set{\barp_r \mid r \in R}$ for each resource role. 
Token moves  from $p_r$ to $\barp_r$  show that the resource gets occupied, and moves from $\barp_r$ to $p_r$ show that the resource becomes available. Also tests whether there are free/occupied resources can be modeled. 
A structural condition is imposed on the net to guarantee that resources are \emph{durable}, meaning that resources can neither be created nor destroyed. This also implies that in the corresponding net system with initial and final marking $m_i$ and $m_f$, $m_i(p_r)=m_f(p_r)$ and $m_i(\barp_r)=m_f(\barp_r)$, for any resource role $r \in R$.

Furthermore, similar to DB-nets and CLog-nets, we extend the tokens from carrying single data values to multiple. Where DB-nets and CLog-nets allow for a variable number of predefined color types, we restrict ourselves to two which are strictly typed, to distinguish between both cases and resources.

\begin{definition}{(Resource-constrained \texorpdfstring{$\nu$}{nu}-Petri net)}\label{def:rc_nunet}
Let $C^\e$ be the set of case ids $\Idc$ extended with ordinary tokens, \ie $\e \in \Idc$, and $\IdR^\e$ be the set of resource ids extended with ordinary tokens. A \emph{resource-constrained $\nu$-Petri net} $N = (P, T, \F, m_i, m_f )$ is a Petri net system with $\F: (P \times T) \cup (T \times P) \rightarrow (\Var^\e_c \times \Var^\e_r)^\oplus$, where $\Var_c$ denote case variables and $\Var_r$ denote resource variables, allowing for two colored tokens. $P = (P_p \uplus P_r \uplus \bar{P}_r)$, with production places $P_p$ and resource availability and busy places $P_r = \set{p_r \mid r \in R}$ and $\bar{P}_r = \set{\barp_r \mid r \in R}$. The following modeling restrictions are imposed on $N$ for each $r \in R$:
\begin{enumerate}
    \item $\pre{p_r} + \pre{\barp_r} = \post{p_r} + \post{\barp_r}$, \ie $\forall_{t \in T}~\F(p_r, t)+\F(\barp_r, t) = \F(t, p_r)+\F(t, \barp_r)$;
    \item $m_i(p_r) = m_f(p_r)$ and $m_i(\barp_r) = m_f(\barp_r) = 0$;
\end{enumerate}

A marking of $N$ is a function $m: P \rightarrow (C^\e \times R^\e)^\oplus$ with case ids $C$ and resources $R$, which is a mapping from places to multisets of colored tokens.

A mode of a transition $t$ is an injection $\mode: (\Var^\e_c \times \Var^\e_r)(t) \rightarrow (C^\e \times R^\e)$, that instantiates each variable to an identifier.
\end{definition}

\begin{proposition}
The resource-constrained $\nu$-Petri nets as defined in Def.~\ref{def:rc_nunet} satisfy requirements R1-R5, 
\ie they allow to \emph{distinguish cases} and \emph{resource instances} which are \emph{durable}, and capture \emph{case-resource correlation} while restricting to operations that are \emph{invertible}.
\end{proposition}
\begin{proof}
The two-colored strictly typed tokens distinguish both the cases (R1) and resource instances (R2) in the system. The modeling restrictions imposed on the RC \nunet~ enforce that for each resource role $r \in R$, tokens can only move between $p_r$ and $\barp_r$, \ie we have the place invariant $(1,1)$ on $p_r$ and $\barp$, implying that $m(p_r) + m(\barp_r) = m_i(p_r)$ for any reachable marking $m$, and that all resource tokens are returned to $p_r$ when the net reaches its final marking, ensuring that resources are durable (R3). The two colors on tokens residing in $\barp$ capture correlation between cases and resources instances (R4), denoting by which case a resource instance is claimed throughout their interaction. As the transition firing's modes are bijective functions, each operation on $N$ is invertible (R5).
\qed
\end{proof}

Note that the RC \nunet~formalism is a restricted version of DB-nets, CLog-net and synchronizing proclets, as all three can capture the behavior that can be modeled by RC $\nu$-nets. DB-nets and CLog-nets additionally have database operations which we deem not relevant for our purposes. Synchronizing proclets allow for many-to-many interactions, while we assume that a resource instance cannot be shared by several cases at the same time.

\section{Complete event logs alignments} \label{sec:log_alignments}
Several state-of-the-art techniques in conformance checking use alignments to relate the recorded executions of a process with a model of this process~\cite{adriansyah2014aligning}. An alignment shows how a log or trace can be replayed in a process model, which can expose deviations explaining either how the process model does not fit reality or how the reality differs from what should have happened. 

Traditionally, this is computed for individual traces, however, as we show in previous work~\cite{sommers2022aligning}, this fails to expose deviations on a multi-case and -resource level in processes with inter-case dependencies as described in Sec.~\ref{sec:rc_nunet}. In this section, we go over the foundations of alignments in Sec.~\ref{sec:al_foundations} and show how we extend this to compute alignments of complete event logs in Sec.\ref{sec:al_extended}.

\subsection{Foundations of alignments}\label{sec:al_foundations}
At the core of alignments are three types of moves: log, model, and synchronous moves (cf. Def.~\ref{def:align_moves}), indicating, respectively, that an activity from the log can not be mimicked in the process model, that the model requires the execution of some activity not observed in the log, and that observed and modeled behavior of an activity agree.
\begin{definition}{(Log, model and synchronous moves)}\label{def:align_moves}
Let $L$ be an event log and $N = (P, T, \F, \ell, m_i, m_f)$ be a labeled $\nu$-Petri net with $T_\mode$ the set of all possible firings in $N$.
We define the set of log moves $\Gamma_l = \set{(e, \gg) \mid e \in L}$, the set of model moves $\Gamma_m = \set{(\gg, t_\mode) \mid t_\mode \in T_\mode}$ and the set of synchronous moves $\Gamma_s = \set{(e, t_\mode) \mid e \in L, t_\mode \in T_\mode, \act(e) = \ell(t)}$. As abbreviations, we write $\Gamma_{ls} = \Gamma_l \cup \Gamma_s$, $\Gamma_{lm} = \Gamma_l \cup \Gamma_m$, $\Gamma_{ms} = \Gamma_m \cup \Gamma_s$, and $\Gamma_{lms} = \Gamma_l \cup \Gamma_m \cup \Gamma_s$.
\end{definition}
Log moves and model moves can expose deviations of the real behavior from the model, by an \emph{alignment} (cf. Def.~\ref{def:alignment}) on a net $(N, m_i, m_f)$ and event log $L$ (possibly a single trace) which is a poset of moves from Def.~\ref{def:align_moves} incorporating the event log and execution sequences in $N$ from $m_i$ to $m_f$:
\begin{definition}{(Alignment)}\label{def:alignment}
An \emph{alignment} $\gamma = \alignm(N, L)$ of an event log  $L=(\bar{L},\prec_L)$ and a labeled Petri net $N = (P, T, \F, \ell, m_i, m_f)$ is a poset $\gamma=(\bar\gamma, \prec_\gamma)$, where $\bar\gamma \subseteq (\Gamma_l \cup \Gamma_s \cup \Gamma_m^\oplus)$, having the following properties:
\begin{enumerate}
    \item $\overline{\gamma\proj_L} = \bar{L}$ and $\prec_L \subseteq \prec_{\gamma\proj_L}$ \label{prop:align1}
    \item $m_i \xrightarrow{\gamma\proj_T} m_f$, \ie $\forall_{\sigma \in (\gamma\proj_T)^<}, m_i \xrightarrow{\sigma} m_f$ \label{prop:align2}
\end{enumerate}
with alignment projections on the log events $\gamma\proj_L$ and on the transition firings $\gamma\proj_{T_\mu}$:
\begin{align}
    \gamma\proj_L &= \left( \sset{e \mid (e, t_\mode) \in \gamma \cap \Gamma_{ls}}, \sset{(e,e') \mid ((e, t_\mode), (e', t'_\mode)) \in \prec_\gamma \cap (\Gamma_{ls} \times \Gamma_{ls})}  \right) \\
    \gamma\proj_T &= \left( \sset{t_\mode \mid (e, t_\mode) \in \gamma \cap \Gamma_{ms}}, \sset{(t_\mode, t'_\mode) \mid ((e, t_\mode), (e', t'_\mode)) \in \prec_\gamma \cap (\Gamma_{ms} \times \Gamma_{ms})}  \right)
    %
    %
\end{align}
\end{definition}

Note the slight difference in the definition of an alignment as opposed to our previous work in \cite{sommers2022aligning}, where the alignment is simplified from a distributed run to a poset of moves.
The process's history of states (markings) as it has supposedly happened in reality can be extracted from the alignment. For the general case, we introduce the \emph{pseudo-firing} of transitions from corresponding alignment's non-log moves in the process model, to obtain a pseudo-marking, which can be unreachable or contain a negative number of tokens:
\begin{definition}{(Pseudo-markings)}\label{def:pseudomarking}
A \emph{pseudo-marking} $m$ of a Petri net $N=(P,T,\F)$ is a multiset 
$P \rightarrow \Int$, \ie the assigned number of tokens a place contains can be negative. $\widetilde\M(N)$ denotes the set of all pseudo-markings in $N$.
\end{definition}

\begin{definition}{(Pseudo-firing of posets)}\label{def:cutmarking}
Let $N=(P,T,\F,m_i,m_f)$ be a RC \nunet~and $\gamma$ be an alignment on $N$. We define a function $\cutmarking: \mathcal{P}(\gamma) \rightarrow \widetilde\M(N)$, with powerset $\mathcal{P}$, to obtain the model pseudo-marking of every subposet of $\gamma$. For every subposet $\gamma' \subseteq \gamma$, we have for every $p \in P$:
\begin{equation}
    \cutmarking( \gamma' )(p) = m_i(p) + \sum_{(e, t_\mode) \in \gamma': t_\mode \neq \epsilon} \left( \mu\left(\F(t,p) \right) - \mu\left(\F(p, t)\right) \right)
\end{equation}
\ie the pseudo-marking is obtained by firing all the transitions of $\gamma'$ with corresponding modes. Note that it is not necessarily reachable.
\end{definition}

An antichain in an alignment denotes a possible point in time, and therefore a state of the process. By pseudo-firing the respective (open) prefix of the antichain, we obtain the corresponding \emph{pre- (or post-)antichain marking}:
\begin{definition}{(Pre- and post-antichain marking)}
Let $\gamma$ be an alignment and $G \in \mathcal{A}(\gamma)$ an antichain in $\gamma$. The pre- (post-)antichain marking defines the marking reached after the pseudo-firing of $\prefixopen{G}$ ($\prefix{G}$), \ie $\cutmarking(\prefixopen{G})$ ($\cutmarking(\prefix{G})$).
\end{definition}

\subsection{Alignments extended to include inter-case dependencies}\label{sec:al_extended}
The foundational work on constructing alignments is presented in~\cite{adriansyah2014aligning} and it relies on the synchronous product of the Petri net $N = (P,T,\F,\ell,m_i,m_f)$ modeling a process and a trace Petri net $N_\sigma = (P^{(\sigma)}, T^{(\sigma)}, \F^{(\sigma)}, \ell^{(\sigma)}, m_i^{(\sigma)}, m_f^{(\sigma)})$ (a Petri net representation of a trace in the event log). The synchronous product consists of the union of $N$ and $N_\sigma$, and a transition $t_s$ for each pair of transitions $(t_m,t_l) \in T \times T^{(\sigma)}$ with $\pre{t_s} = \pre{t_m} + \pre{t_l}$ and $\post{t_s} = \post{t_m} + \post{t_l}$, iff $t_m$ and $t_l$ share the same label and variables on the incoming arcs, \ie $\ell(t_m) = \ell^{(\sigma)}(t_l)$ and $Var(t_m) = Var(t_l)$. The alignment is then computed by a depth-first search on the synchronous product net from $m_i + m_i^{(\sigma)}$ to $m_f + m_f^{(\sigma)}$ using the $A^*$ algorithm, with the firings of transition from $T^{(\sigma)}$, $T$ and $T^{(s)}$ corresponding to the log, model and synchronous moves from Def.~\ref{def:align_moves}~\cite{adriansyah2014aligning}.

With $c: \Gamma_{lms} \rightarrow \mathbb{R}^+$ a cost function, usually defined for each $(e, t_\mode) \in \Gamma_{lms}$ as follows:
\begin{equation}
    c((e, t_\mode)) =
    \begin{cases}
        0        & (e, t_\mode) \in \Gamma_s \\
        1        & (e, t_\mode) \in \Gamma_{lm} \wedge \ell(t) \neq \tau \\
        \epsilon & \ell(t) = \tau
    \end{cases} 
\end{equation}
The \emph{optimal alignment} is an alignment $\gamma$ such that $\sum_{g \in \gamma} c(g) \leq \sum_{g \in \gamma'} c(g)$ holds for any alignment $\gamma'$, which prefers synchronous moves over model and log moves. In terms of conformance checking and exposing realistic deviations, the optimal alignment provides the ``best'' explanation for the relation between observed and modeled behavior.

In Sec.~\ref{sec:rc_nunet}, we have shown how a RC \nunet~is a Petri net formalism with capability of modeling inter-case dependencies and suitability for conformance checking. We extend the alignment problem 
  in order to expose inter-case deviations by adapting the synchronous product net to \nunet s: an RC \nunet~and the log \nunet:

\begin{definition}{(Log $\nu$-Petri net)}\label{def:trace_nu_net}
Given an event log $L$, a \emph{log $\nu$-Petri net} $N^{(L)} = (P^{(L)}, T^{(L)}, \F^{(L)}, \ell^{(L)}, m_i^{(L)}, m_f^{(L)})$ is a labeled \nunet~constructed as follows.
For every $e \in L$, we make a transition $t_e \in T^{(L)}$ with $\ell^{(L)}(t) = \act(e)$, and for each resource instance $\idr_r \in supp(\Res(e))$ we make a place $p \in P^{(L)}$ with $\pre{p} = \emptyset$, $\pre{p} = \mset{|\idr_r| \cdot t}$, $\F^{(L)}(p, t) = \mset{|\idr_r| \cdot (\e, r)}$ and $m_i^{(L)}(p)((\e, \idr)) = |\idr|$.
Further, for every pair $(e_1, e_2) \in \prec_L$, we make a place $p \in P^{(L)}$ with $\pre{p} = \mset{t_{e_1}}$, $\post{p} = \mset{t_{e_2}}$ and 
\begin{equation}
    \F^{(L)}(t_{e_1}, p) = \F^{(L)}(p, t_{e_2}) =
    \begin{cases}
        \mset{(c, \e)}           & \ecase(e_1) = \ecase(e_2) \\
        \mset{(\e, \e)} & \text{otherwise}
    \end{cases}
\end{equation}
For every $e^{-} \in \min(L)$, we make a place $p^{-} \in P^{(L)}$ with $\pre{{p^{-}}} = \emptyset$, $\post{{p^{-}}} = \mset{t_{e^{-}}}$ and $m_i^{(L)}(p^{-})((case(e^{-}), \e)) = 1$. 
Similarly, for every $e^{+} \in \max(L)$, we make a place $p^{+} \in P^{(L)}$ with $\pre{{p^{+}}} = \mset{t_{e^{+}}}$, $\post{{p^{+}}} = \emptyset$ and $m_f^{(L)}(p^{+})((case(e^{+}), \e)) = 1$.
\end{definition}

Computing the \emph{complete event log alignment} is again a matter of finding a path from the initial to the final marking in the synchronous product net, \ie from $m_i + m_i^{(L)}$ to $m_f + m_f^{(L)}$, for which we can use \emph{any} of the existing methods as described before. The optimal alignment is again the one with lowest cost. In terms of complexity, the alignment problem with an empty event log and an all-zero cost function can be reduced to the reachability problem for bounded Petri nets from $m_i$ to $m_f$, which has exponential worst-case complexity\cite{murata1989petri}. Adding event to the log $\nu$-Petri net and a non-zero costs on moves makes the problem strictly more complex.

Note that while $\nu$-Petri nets are inherently unbounded in general due to the generation of fresh tokens, we can retain boundedness in the context of alignments, since the bound is predicated by the event log and we can get this information by preprocessing it.

\begin{figure}
\centering
\includegraphics[scale=0.5]{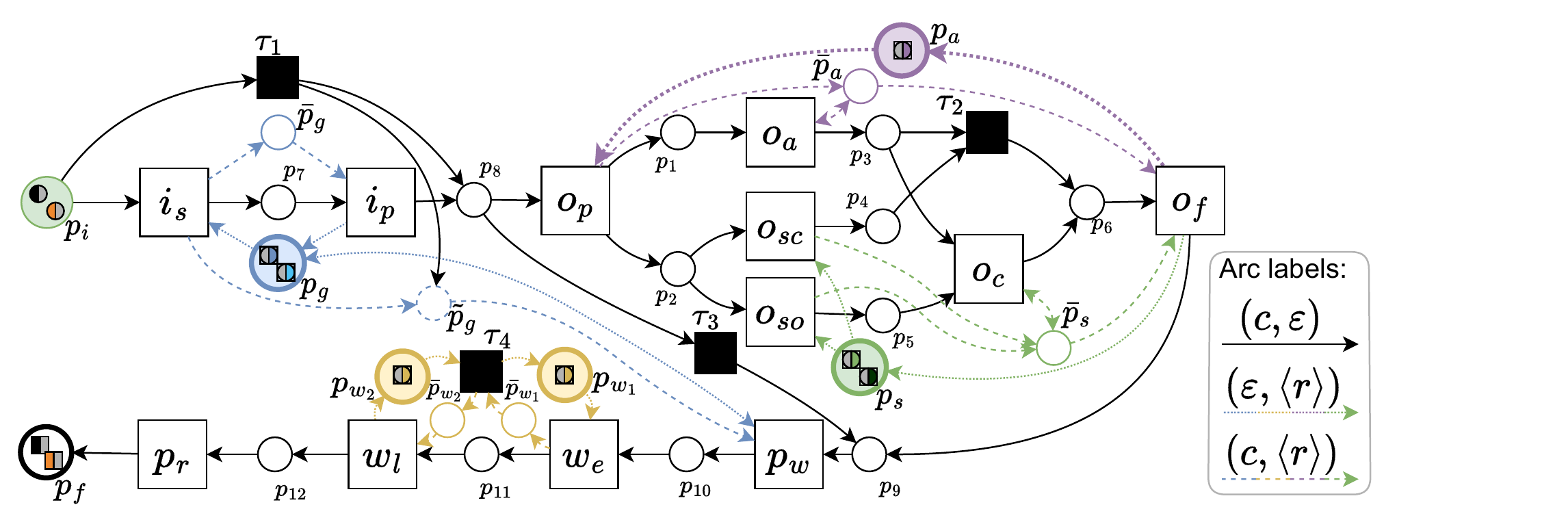}
    \caption{Process model RC \nunet~$N$, with initial and final marking, annotated with circular and square tokens respectively.}
    \label{fig:re_net}
\includegraphics[scale=0.5]{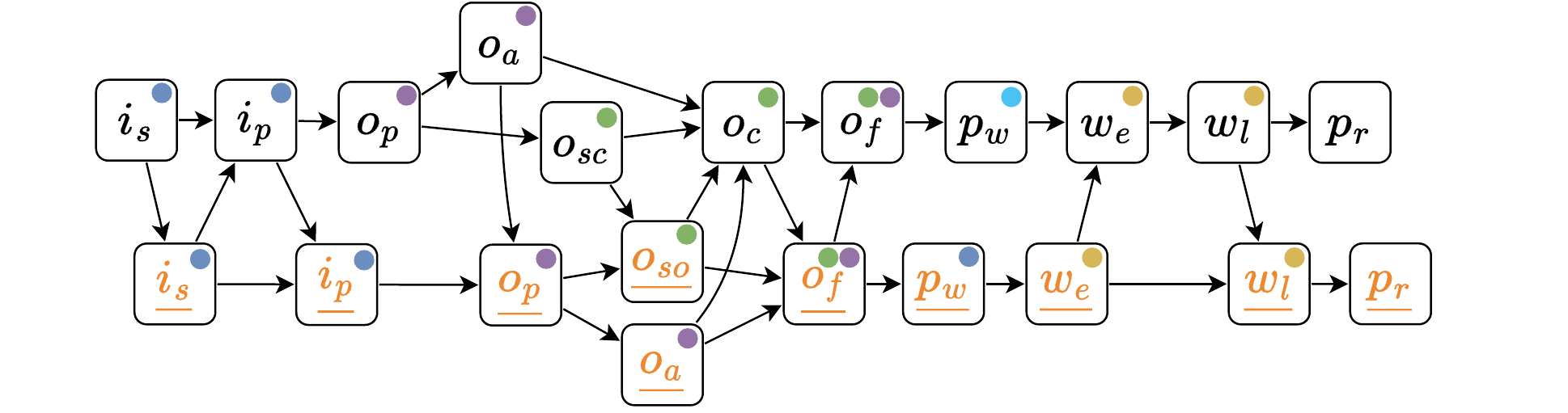}
    \caption{Event log $L$.}
    \label{fig:re_log}
\includegraphics[scale=0.5]{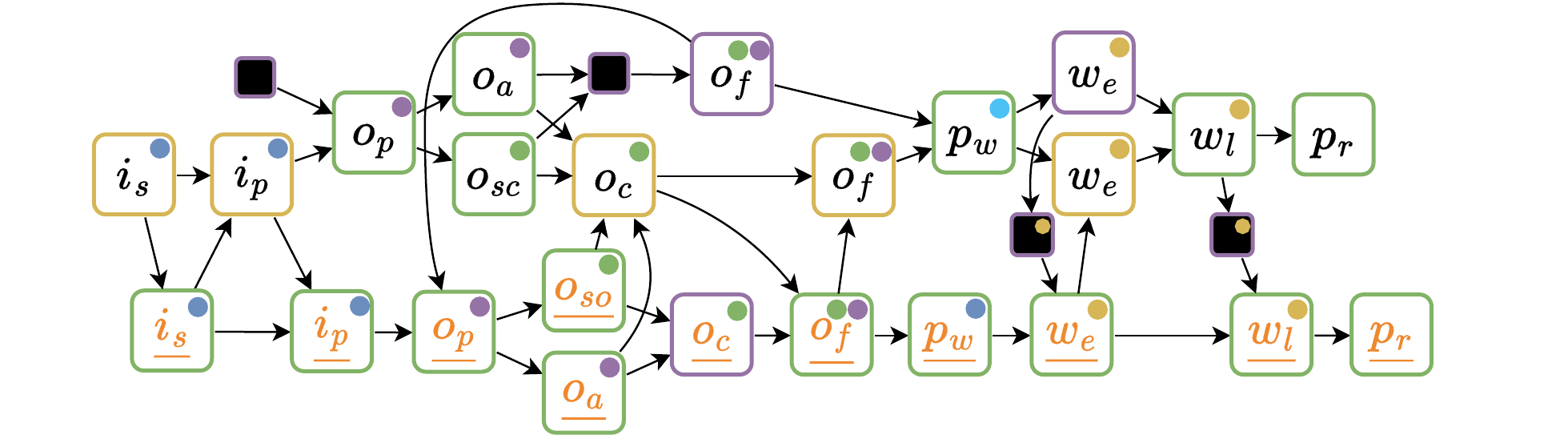}
    \caption{Complete event log alignment $\gamma$, with the colors depicting the move types; green, purple, and yellow for synchronous, model, and log moves respectively.}
    \label{fig:re_alignment}
\includegraphics[scale=0.5]{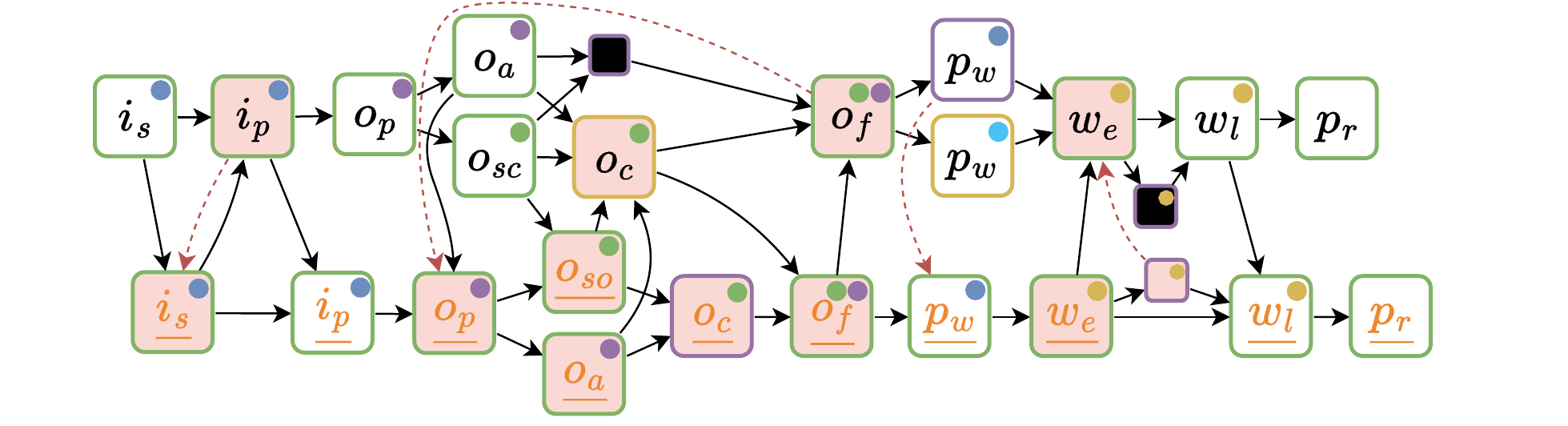}
    \caption{Composed alignment $\compal$ with annotated permutation and realignment intervals.}
    \label{fig:re_compal}
\includegraphics[scale=0.5]{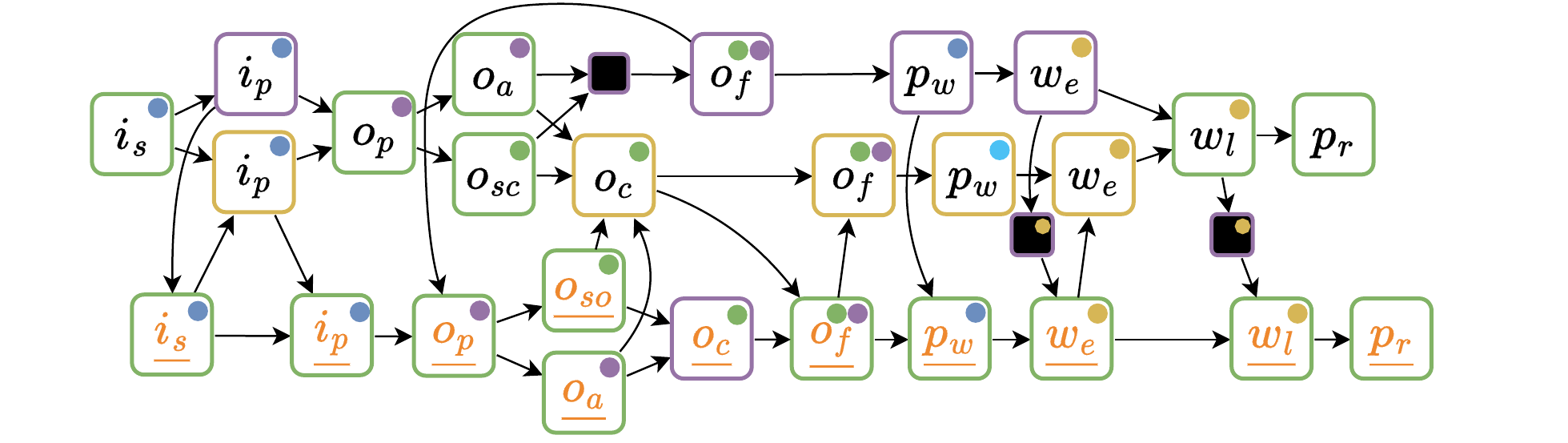}
    \caption{Approximated alignment $\gamma^*$.}
    \label{fig:re_alignment_approx}
\end{figure}
For our running example, modeled in Fig.~\ref{fig:re_net}, we extend the small operation process from Fig.~\ref{fig:requirements} with an assistant resource during the operation, an intake subprocess $(i_s, i_p)$ involving a general practitioner (GP), and a prescription subprocess with a FIFO waiting room $(p_w, w_e, w_l, p_r)$, where the prescription can only be written by the GP involved in the intake, if appropriate. Both the intake and operation subprocesses can be skipped via silent transitions $\tau_1$ and $\tau_3$ respectively in $N$. Fig.~\ref{fig:re_log} shows the recorded event log $L$ of this process which concerns two patients. An optimal complete event log alignment on $N$ and $L$, computed by the method above is presented in Fig.~\ref{fig:re_alignment}.

\section{Approximation by composition and local realignments}\label{sec:approximation}
Since multiple cases are executed in parallel, computing the alignment on the complete event log $L$, as described in Section~\ref{sec:log_alignments}, is a computationally expensive task. At the same time, one can see that the multi-case and -resource alignment only deviates from the classical individual alignments when violations occur on the inter-case dependencies, \eg when a resource is claimed while it is already at maximal capacity.
\begin{figure}[t]
\centering
\includegraphics[scale=0.5]{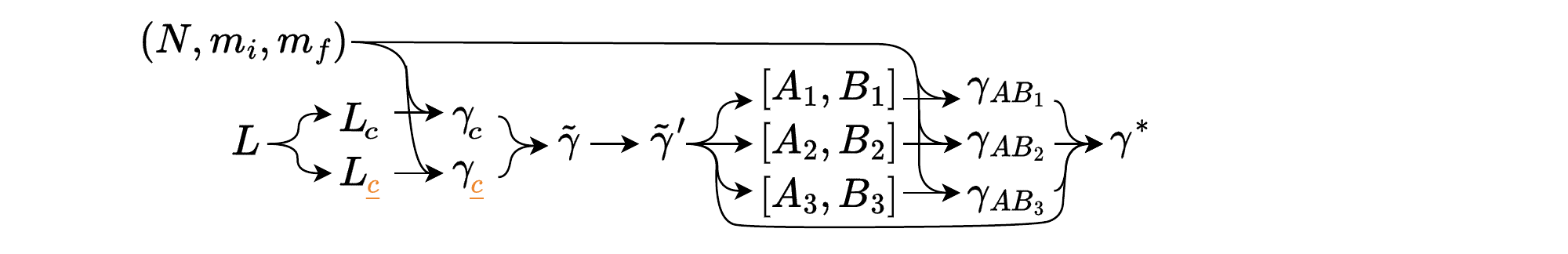}
    \caption{Overview of our approximation method.}
    \label{fig:approximation_overview}
\end{figure}

We can approximate the alignment of a complete event log $L$ and a Petri net $N$ by using a composition of individually aligned cases. An overview of this method is illustrated in Fig.~\ref{fig:approximation_overview}, which we subdivide into two parts, described respectively in Sec.~\ref{sec:comp_viol} and \ref{sec:resolve_vio}.
\begin{enumerate}
    \item $L$ is decomposed into the individual cases ($L_c, L_{\co{c}}$), which are aligned to $N$ ($\gamma_c, \gamma_{\co{c}}$) and composed using the event log's partial order $\prec_L$ ($\compal$). The result is not necessarily an alignment as inter-case deviations may be left unresolved;
    \item We transform this composed alignment into a valid alignment by taking a permutation ($\compal'$) and realigning parts ($\interval{A_1}{B_1}, \interval{A_2}{B_2}, \interval{A_3}{B_3}$) of the event log locally to resolve the violations. The approximated alignment ($\gamma^*$) is obtained by substituting the realignments ($\gamma_{AB_1}, \gamma_{AB_2}, \gamma_{AB_3}$). 
\end{enumerate}
%

The implementation of both the original method from~\cite{sommers2022aligning} and the approximation method for computing complete event log alignments is available at \href{https://gitlab.com/dominiquesommers/mira}{gitlab.com/dominiquesommers/mira}, including the examples used in this paper and some additional examples.

\subsection{Composing individual alignments}\label{sec:comp_viol}
For every case $c \in \Idc$, we have the trace $L_c$ (cf. Def.~\ref{def:event_log}) projected on the case identifier $c$. As described in Sec.\ref{sec:log_alignments}, the optimal complete event log alignment $\gamma_L$ consists of individual alignments $\gamma_c$, on $N$ and $L_c$ for every $c \in \Idc$, composed together respecting the event log's partial order $\prec_L$, where each $\gamma_c$ is not necessarily optimal with regard to $L_c$.

It is computationally less expensive to compute the optimal alignments $\gamma_c = \alignm(N, L_c)$  for each $c \in \Idc$ and then approximate $\gamma_L$. We create a composed alignment $\compal$ with the optimal individual alignments and the event log's partial order, as defined in Def.~\ref{def:composed_al}. Fig.\ref{fig:re_compal} shows the composed alignment for the running example with additional annotations (in red) which we cover later.
\begin{definition}{(Composed alignment)}\label{def:composed_al}
Given a Petri net $N$ and an event log $L$ with traces $L_c$ for $c \in \Idc$, let $\gamma_c = \alignm(N, L_c)$ be the corresponding optimal individual alignments. The \emph{composed alignment} $\compal = \doublecup_{\idc \in \Idc} \gamma_c$ is the union of individual alignments with the extended partial order on the synchronous moves, defined as the transitive closure of the union of partial orders from the individual alignments and the partial order on moves imposed by the partial order $\prec_L$ of the event log:

\begin{equation}
    \prec_\compal \;= \left( \bigcup_{c \in \Idc} \prec_{\gamma_c} \cup \prec_{\gamma_L} \right)^+
\end{equation}
with $\prec_{\gamma_L} = \left\{ \spair{(e, t_\mode)}{(e', t'_\mode)} \mid e \prec_L e', (e, t_\mode), (e', t'_\mode) \in (\gamma \cap \Gamma_{ls}) \right\}$. 
\end{definition}
Recall that for every sequence $\sigma \in \compal\proj_T$ of an alignment $\compal$, we have $m_i \xrightarrow{\sigma} m_f$, \ie $\sigma$ is a firing sequence in $N$. This property is not guaranteed for a composed alignment, even in the absence of inter-case deviations. In the presence thereof, we say that a composed alignment is \emph{violating} as there exists no such sequence.
\begin{definition}{(Violating composed alignment)}\label{def:violations}
Let $\idr_r \in \supp(\IdR)$ be a resource instance and $\compal = \doublecup_{\idc \in \Idc} \gamma_c$ a composed alignment. We define 
\begin{equation}
    \mathcal{S}(\compal) = \set{(\bar\compal', \prec_{\compal'}) \mid \bar\compal = \bar\compal', \prec_\compal \subseteq \prec_{\compal'}, \prec_{\compal'} = (\prec_{\compal'})^+, \forall_{g \in \compal} g \nprec_{\compal'} g }
\end{equation}
as the set of transitively closed and acyclic antichain permutations of $\compal$ that respect the partial order $\prec_\compal$.

$\compal$ is in violation with any of the resource instances if and only if:
\begin{equation}\label{eq:violation}
    \forall_{\compal' \in \mathcal{S}(\compal)} \exists_{G \in \mathcal{A}^{+}(\compal')} \viol(G)
\end{equation}
with violation criteria $\viol: 
\mathcal{A}^+(\compal)
\rightarrow \mathbb{B}$ defined for each maximal antichain $G \in \mathcal{A}^+(\compal)$ 
as follows:
\begin{equation}
    \viol(G) = \exists_{\idr_r \in supp(\IdR)} \left[ \cutmarking(\prefixopen{G})(p_r)((\e,\idr_r)) < \sum_{(e,t_\mode)\in G} \F(p_r,t)(\mode^{-1}((\e, \idr_r)))\right]
\end{equation}
\ie there is no way of firing all transitions in the alignment such that at all times enough capacity is available.
\end{definition}
In Fig.~\ref{fig:re_compal}, antichains meeting the violation criteria are the single moves with an incoming red arc. In Theorem~\ref{the:violating_not_firable} we show that for every sequence of transitions $\sigma \in \compal\proj_T$ in violating composed alignment $\compal$, we have $m_i \not\xrightarrow{\sigma} m_f$, \ie $\compal$ is not firable.
\begin{theorem}{(A violating composed alignment is not firable)}\label{the:violating_not_firable}
Let $\compal = \doublecup_{\idc \in \Idc} \gamma_c$ be a composed alignment on RC \nunet~$N = (P,T,\F,m_i,m_f)$ and event log $L$, such that $\compal$ is violating. Then there exists no firing sequence $\sigma$ in $\compal$ such that $m_i \xrightarrow{\sigma} m_f$.
\end{theorem}
\begin{proof}
$\compal$ is violating, therefore, for every $\compal' \in \mathcal{S}(\compal)$, there is a maximal antichain $G \in \mathcal{A}^{{+}}(\compal')$ and resource instance $\idr_r \in supp(\IdR)$, such that
\begin{align}
    &\cutmarking(\prefixopen{G})(p_r)((\e,\idr_r)) < \sum_{(e,t_\mode)\in G} \F(p_r,t)(\mode^{-1}((\e, \idr_r))) \\
    & \cutmarking(\prefixopen{G})(p_r)((\e,\idr_r))\: - \sum_{(e,t_\mode)\in G} \F(p_r,t)(\mode^{-1}((\e, \idr_r))) < 0
\end{align}
hence firing the transitions in $G$ leads to a negative marking for $(\e,\idr_r)$ in place $p_r$, which is invalid.\qed
\end{proof}
With an antichain $G \subseteq \mathcal{A}(\compal)$, we show in Lemma~\ref{lem:antichain_marking_viol} that $\cutmarking(\prefixopen{G})$ (and $\cutmarking{\prefix{G}}$) is reachable if an only if the prefix $\prefixopen{G}$ ($\prefix{G}$) is not violating.
\begin{lemma}{(A pre- (and post-)antichain marking in a composed alignment is reachable iff the corresponding prefix is not violating.)}\label{lem:antichain_marking_viol}
Let $\compal = \doublecup_{\idc \in \Idc} \gamma_c$ be a composed alignment on RC \nunet~$N = (P,T,\F,m_i,m_f)$ and event log $L$ and let $G \in \mathcal{A}(\compal)$ be an antichain. Then the pre- (and post-)antichain marking $\cutmarking(\prefixopen{G})$ ($\cutmarking(\prefix{G})$) is reachable if and only if $\prefixopen{G}$ ($\prefix{G}$) is not violating.
\end{lemma}
\begin{proof}
We prove the lemma by proving both sides of the bi-implication:

$(\implies)$ $m_G = \cutmarking(\prefixopen{G})$ is reachable, hence there exists a sequence $\sigma \in \prefixopen{G}^*$ with $\prec_{\prefixopen{G}} \subseteq \prec_\sigma$ such that $m_i \xrightarrow{\sigma} m_G$. Let $\compal' \in \mathcal{S}(\compal)$ be an antichain permutation with $\prec_\sigma \subseteq \prec_{\compal'}$. Then by definition of reachable marking, for every maximal antichain $G \in \mathcal{A}^+(\compal')$ and every resource instance $\idr_r \in supp(\IdR)$, we have $\cutmarking(\prefixopen{G})(p_r)((\e, \idr_r)) \geq \sum_{(e, t_\mode) \in G} \F(p_r,t)(\mode^{-1}((\e, \idr_r)))$. Thus $\prefixopen{G}$ is not violating.

$(\impliedby)$ $\prefix{G}$ is not violating, hence there exists a $\compal' \in \mathcal{S}(\prefix{G})$, such that for all $G' \in \mathcal{A}^{+}(\compal')$ and all $\idr \in supp(\IdR)$ we have:
\begin{equation}
    \cutmarking(\prefixopen{G})(p_r)((\e,\idr_r)) \geq \sum_{(e,t_\mode)\in G} \F(p_r,t)(\mode^{-1}((\e, \idr_r)))
\end{equation}
$m_i \xrightarrow{\sigma} \cutmarking(\prefix{G})$ with $\sigma$ respecting the partial order $\prec_{\compal'}$. \qed
\end{proof}

\subsection{Resolving violations in the composed alignment}\label{sec:resolve_vio}
Let $\compal' \in \mathcal{S}(\gamma)$ be an antichain permutation of $\compal$. Then, by Def.~\ref{def:violations}, we have a set of violating maximal antichains (which is empty when $\compal$ is not violating) where the corresponding transitions are not enabled. Instead of needing to align the complete event log, we show that we can resolve violations locally around such antichain. For each violating antichain $G$, there exists an interval $\interval{A}{B} \subseteq \compal'$ with $A \preceq G \preceq B$ such that $\interval{A}{B}$ is alignable, formally defined in Def.~\ref{def:align_interval}.
\begin{definition}{(Alignable interval)}\label{def:align_interval}
    Let $\gamma = \doublecup_{\idc \in \Idc} \gamma_c$ be a composed alignment on RC \nunet~$N = (P,T,\F,m_i,m_f)$ and event log $L$, and let $A,B \in \mathcal{A}(\gamma)$ be two antichains. We say that the interval $\interval{A}{B}$ is alignable if and only if $m_B = \cutmarking(\prefix{B})$ is reachable from $m_A = \cutmarking(\prefixopen{A})$, \ie $m_A \xrightarrow{*} m_B$, assuming $m_A$ is reachable.
\end{definition}

Note that $\interval{\min(\gamma')}{\max(\gamma')}$ is always an alignable interval. We use our running example to show that it can be taken locally around $G$ instead, \eg $\interval{\set{\co{i_s}}}{\set{i_p}}$ with $G = \set{\co{i_s}}$ (cf.~ Fig.~\ref{fig:re_compal}). Note how the violation can be resolved by substituting $\interval{A}{B}$ by a subalignment from $m_A = \cutmarking(\prefixopen{A})$ to $m_B = \cutmarking(\prefix{B})$.

In order to prove statements that do not depend on a chosen realignment mechanism, we now assume that there exists a function $f_\compal: \mathcal{A}^+(\compal) \rightarrow \mathcal{P}(\compal)$ that produces an alignable interval $\interval{A}{B}$ for an arbitrary $G \in \mathcal{A}^+(\compal)$. 

\begin{align}
    W(\compal'_V) = \{ \interval{\min(\gamma_v)}{\max(\gamma_v)} \mid &~\gamma_v \subseteq \compal'_V, \forall_{g \in \gamma_v, g' \in \compal'_V \setminus \gamma_v} g \|_{\compal'_V} g', \\
    &~\forall_{g \in \gamma_v}\exists_{g' \in \gamma_v} g \not\|_{\gamma_v} g' \} \nonumber
\end{align}
with $\compal'_V = \bigcup_{G \in \mathcal{A}^+(\compal')} f_{\compal'}(G)$, denotes the set of alignable intervals covering every violating antichain in $\compal'$, and it is annotated in red for the running example in Fig.~\ref{fig:re_compal}, with the three intervals $\interval{\set{\co{i_s}}}{\set{i_p}}$, $\interval{\set{\co{o_p}}}{\set{o_f}}$, and $\interval{\set{\co{w_e}}}{\set{\tau}}$ covering the violating antichains $\set{\co{i_s}}$, $\set{\co{o_p}}$, $\set{\co{o_{so}}}$, and $\set{w_e}$.

We resolve the violations in $\compal'$ by substituting every interval $\interval{A}{B} \in W(\compal'_V)$ by an alignment $\gamma_{AB}$ on $N$ and $\interval{A}{B}\proj_L$ from $m_A = \cutmarking(\prefixopen{A})$ to $m_B = \cutmarking(\prefix{B})$.

Since, for now, we assume that every interval $f(G)$ is alignable, a subalignment $\gamma_{AB}$ exists. The approximated alignment $\gamma^* = (\bar\gamma^*, \prec_{\gamma^*})$ is then defined as follows:
\begin{align}
    \bar\gamma^* &= \bigcup_{[A,B] \in W(\compal'_V)} \bar\gamma_{AB} \cup (\bar\compal \setminus \bar\compal'_V) \label{eq:al_approx_set}\\
    \prec_{\gamma^*} &= \left( \bigcup_{[A,B] \in W(\compal'_V)} \prec_{\gamma_{AB}} \cup \sset{(g_1, g_2) \mid g_1,g_2 \in \gamma \setminus \compal'_V, g_1 \prec_{\compal'} g_2} \right)^+   \label{eq:al_approx_prec}
\end{align}
$\gamma^*$ for the running example is shown in Fig.~\ref{fig:re_alignment_approx} with substituted realignments for the intervals annotated in red from Fig.~\ref{fig:re_compal}. Note that $\gamma^*$ is an approximation of the optimal alignment $\gamma$ from Fig.~\ref{fig:re_alignment} as $c(\gamma^*) \geq c(\gamma)$, due to the local realignments. In Theorem~\ref{the:ilp_alignment} we show that $\gamma^*$ is a valid alignment.
\begin{theorem}{($\gamma^*$ is an alignment.)}\label{the:ilp_alignment}
Let $\compal = \doublecup_{\idc \in \Idc} \gamma_c$ be a composed alignment on RC \nunet~$N = (P,T,\F,m_i,m_f)$ and event log $L$ and let $\compal' \in \mathcal{S}(\compal)$ be an antichain permutation of $\compal$, with $W(\compal'_V)$ the set of alignable intervals covering every violating antichain in $\compal'$.

$\gamma^* = (\bar\gamma^*, \prec_{\gamma^*})$, following Eqs.~\ref{eq:al_approx_set} and \ref{eq:al_approx_prec}, is a valid alignment, \ie it has properties (1), (2) and (3) from Def.~\ref{def:alignment}.
\end{theorem}
\begin{proof}
We prove that $\gamma^*$ is an alignment by induction on the size of $W(\compal'_V)$. For the base case with $|W(\compal'_V)|=0$, we have $\bar\gamma^* = \bar\compal$ and $\prec_{\gamma^*} = \prec_{\compal'}$. By definition, $\bar{\compal\proj_L} = \bar{L}$ and $\prec_L \subseteq \prec_{\compal\proj_L}$. Furthermore, since $|W(\compal'_V)|=0$, we know that for all $G \in \mathcal{A}^+(\compal')$, we have $\neg \viol(G)$, implying that $m_i \xrightarrow{\compal'} m_f$.

Let us assume that $\gamma^*$ is an alignment for $|W(\compal'_V)| = w$. We prove the statement for $W'(\compal'_V) = W(\compal'_V) \cup \set{\interval{A}{B}}$ with $|W'(\compal'_V)| = w + 1$ and $\interval{A}{B} \in \min(W'(\compal'_V))$. For every maximal antichain $G \in \mathcal{A}^+(\prefixopen{A})$ before $A$, \ie $G \prec A$, we have $\neg \viol(G)$, which we prove by contradiction. Assume $\viol(G)$, then by our assumption of the existence of $f_{\compal'}$, there is an alignable interval $\interval{A'}{B'} \subseteq \compal'$ with $A' \preceq G \preceq B'$, thus, by $G \prec A$, we have $\interval{A'}{B'} \prec \interval{A}{B}$, implying that $\interval{A}{B} \notin \min(W'(\compal'))$ which is a contradiction. By Lemma~\ref{lem:antichain_marking_viol} and the assumption that $f_{\compal'}(G)$ is an alignable interval, $m_i \xrightarrow{*} m_A \xrightarrow{*} m_B$ and $\interval{A}{B}$ can be  substituted by $\gamma_{AB}$ without violations in $\prefix{B}$, completing the proof. \qed
\end{proof}

\subsection{Obtaining minimal local alignable intervals}
We propose a method to find an antichain permutation of a composed alignment $\compal$ together with the intervals $W(\compal_V)$ such that all violations can be resolved by realigning these intervals as described in Sec.~\ref{sec:resolve_vio}. For computational efficiency, we choose to minimize the number of moves in the intervals that need to be realigned.

We formulate this as an Integer Linear Programming (ILP) problem. The objective of the ILP problem is to adjust the partial order of $\compal$, such that alignable intervals can be identified around violating antichains, preferring intervals with fewer moves. 

Let there be a (possibly arbitrary) fixed order in $\compal$ and $\IdR$ such that each element has a unique index, \ie for every $1 \leq i \leq n_\compal$, $\compal(i)$ and $(e(i), t_\mode(i))$ both denote the $i^{\text{th}}$ move in $\compal$, with $n_\compal = |\compal|$. Furthermore, for every $1 \leq j \leq n_r$, $\IdR(j)$ denotes the $j^{\text{th}}$ resource instance, with $n_r = |supp(\IdR)|$.

Let $\matr{R}$ be a $n_\compal \times n_\compal$ matrix, with $\matr{R}$ defined for every two indices $1 \leq i,j \leq n_\compal$ such that $\matr{R}_{ij}$ is a binary value denoting $(\compal(i), \compal(j)) \in \prec_{\compal}$. For each $\idc \in \Idc$, we introduce  the set $I_\idc$ of indices corresponding to moves in $\compal\proj_{\gamma_c}$. Furthermore, we use $[1..n] = \set{1,\dots,n}$ as an abbreviation for the set of all indices from 1 to $n$.

The set of minimal alignable intervals containing all violations, denoted by $W(\compal'_V)$, with $\compal'_V$ given by
\begin{equation}
    \compal'_V = \bigcup_{i,j \in [1..n
    _\compal]:\matr{X}_{ij} - \matr{R}_{ij}=1} \interval{\compal(j)}{\compal(i)}
\end{equation}
where $\matr{X}$ denotes the new partial order relation between alignment moves which respects the resources capacities and provides the solution to
\begin{equation}
    \text{Minimize } \sum_{i,j \in [1..n_\compal]} (1 - \matr{R}_{ij})\matr{R}_{ji}\matr{X}_{ij} + \epsilon \cdot (1 - \matr{R}_{ij})(1 - \matr{R}_{ji})\matr{X}_{ij}
\end{equation}
subject to
\begin{align}
    & \forall_{i,j \in [1..n_\compal]} & \matr{X}_{ij} & \in \set{0,1} \\
    %
    %
    %
    & \forall_{\idc \in \Idc} \forall_{i,j \in I_c} & \matr{X}_{ij} & = \matr{R}_{ij} \label{eq:const_c} \\
    & \forall_{i,j \in [1..n_\compal]} & \matr{R}_{ij} + (1 - \matr{X}_{ij}) - \matr{X}_{ji} & \leq 1 \label{eq:const_rev_rem} \\
    & \forall_{i,j,k \in [1..n_\compal]} & \matr{X}_{ij} + \matr{X}_{jk} - \matr{X}_{ik} & \leq 1 \label{eq:const_trans_clos} \\
    & \forall_{i \in [1..n_\compal]} & (1 - \matr{X}_{i \bullet}) \matr{C}^\clm - \matr{X}_{\bullet i}^T \matr{C}^\rls & \leq \vect{k} \label{eq:const_vio}
\end{align}
with $\matr{C}^\clm$ and $\matr{C}^\rls$ both $n_\compal \times n_r$ matrices counting how many resource instances are claimed and released respectively  for every alignment move. Both are defined for every $i \in [1..n_\compal]$ and $k \in [1..n_r]$ with $(e, t_\mode) = \compal(i)$ and $\idr_r = \IdR(k)$:
\begin{equation}
    \matr{C}^\clm_{ik} = \F(p_r,t)((\varepsilon, \mode^{{-1}}(\idr_r))) \text{ and }
    \matr{C}^\rls_{ik} = \F(t,p_r)((\varepsilon, \mode^{{-1}}(\idr_r)))
\end{equation}
and capacity vector $\vect{k}$ of length $n_r$, defined as $\vect{k}_k = |\IdR(k)|$ for every $k \in [1..n_r]$.

$\matr{X}$ provides the solution of a new partial order of moves in $\compal$ such that all violations are resolved and the least number of partial order relations is removed. For the running example, the additional arcs from the solution $\matr{X}$ are shown in red in Fig.~\ref{fig:re_compal}.

We refer to App.
\if \arxiv1
    \ref{app:ilp_correctness}
\else
    A in~\cite{sommers2023efficiently}
\fi
for the correctness proof of the ILP problem, where we show (1) the effectiveness of each constraint, (2) that there always exists a solution, (3) that the optimal solution has zero cost if and only if the composed alignment is not violating, and (4) that each interval obtained in $W(\compal'_V)$ is alignable.

\section{Conclusion}\label{sec:conclusion}
We have formulated the requirements for modeling and analyzing processes with inter-case dependencies and argued that our previously proposed Petri net extension named Resource Constrained $\nu$-Petri nets meets them. This paper continues on work presented in~\cite{sommers2022aligning}, where we showed that the traditional methods of aligning observed behavior with the modeled one fall short when dealing with coevolving cases, as they consider isolated cases only. The technique we present here aligns multiple cases simultaneously, exposing violations on inter-case dependencies. We developed and implemented an approximation technique based on a composition of individual alignments and local resolution of violations, which is an important advancement for the use of the technique in practice.

There can be ambiguity in the interpretation of the exposed violations, \eg was the activity executed but not recorded, executed by an ``incorrect'' resource instance, or not executed at all? In~\cite{sommers2022aligning}, we briefly touched upon relaxations of the synchronous product model as a means to improve the deviations' interpretability. One such relaxation helps to detect situations when a step required by the model was skipped in a process execution, and the resources needed for the step were not available at the time when it should have been executed. Adding ``resource-free'' model moves for  transitions allows to capture such deviations. Such special moves, when present in the alignment, reduce the ambiguity and provide a better explanation, \eg that the activity was not executed at all, rather than it might also have been executed but not recorded. For future work, we plan to extend and formalize the relaxations, and evaluate the insights obtained with the alignments based on a real-life case study.

\subsubsection{Acknowledgments.}
This work is done within the project ``Certification of production process quality through Artificial Intelligence (CERTIF-AI)'', funded by NWO (project number: \href{https://www.nwo.nl/projecten/17998}{17998}).

\bibliographystyle{plain}
\bibliography{references}

\newpage
\appendix

\if \arxiv1
    \section{Correctness of the ILP problem}\label{app:ilp_correctness}
We first show the effectiveness of each constraint:
\begin{itemize}
    \item Constraint~\ref{eq:const_c} ensures that the original partial order of the individual alignments is preserved. Note that this also ensures $g \nprec g$ for every $g \in \compal$;
    \item Const.~\ref{eq:const_rev_rem} enforces that when a relation $g \prec g'$ is removed, the opposite $g' \prec g$ is added:
    \begin{equation}
        \matr{R}_{ij} = 1 \wedge \matr{X}_{ij} = 0 \implies \matr{X}_{ij} = 1
    \end{equation}
    \item Const.~\ref{eq:const_trans_clos} enforces that the transitive closure is covered:
    \begin{equation}
        \matr{X}_{ij} = \matr{X}_{jk} = 1 \implies \matr{X}_{ik} = 1
    \end{equation}
    With Const.~\ref{eq:const_trans_clos} together with $g \nprec g$ for every $g \in \compal$ from Const.~\ref{eq:const_c}, there can be no loops in the solution of the ILP problem;
    \item Const.~\ref{eq:const_vio} enforces that the solution is not violating with regard to any resource instance capacities, which we show in Lemma~\ref{lem:const_vio}.
\end{itemize}
\begin{lemma}{(Constraint~\ref{eq:const_vio} ensures there are no violations)}\label{lem:const_vio}
With a composed alignment $\compal = \doublecup_{\idc \in \Idc} \gamma_c$ and $\matr{X}$ a solution to the ILP problem formulated above. The permutation $\compal' = (\bar\compal', \prec_{\compal'})$ of $\compal$ following the partial order of $\matr{X}$, \ie $\bar\compal' = \bar\compal$ and $\prec_{\compal'} = \set{ (\compal'(i), \compal'(j)) \mid i,j \in [1..n_\compal], \matr{X}_{ij} = 1}$, is not violating (cf. Def.~\ref{def:violations}).
\end{lemma}
\begin{proof}
First, let us rewrite Const.~\ref{eq:const_vio} to
\begin{equation}\label{eq:rewrite_const_vio}
    \forall_{i \in [1..n_\compal]}\forall_{k \in [1..n_r]}(1 - \matr{X}_{i\bullet}) \matr{C}^\clm_{\bullet k} - \matr{X}^T_{i\bullet} \matr{C}^\rls_{\bullet k} \leq \vect{k}_k = |\IdR(k)| 
\end{equation}

For ever move index $i \in [1..n_\compal]$, we can define $G_i = \set{ \compal(j) \mid j \in [1..n_\compal], \matr{X}_{ij} = \matr{X}_{ji} = 0 }$ to be the maximal antichain in $\matr{X}$ that contains $\compal(i)$. $(1 - \matr{X}_{ij})$ and $\matr{X}^T_{ij})$ relate to $G_i$ as follows:
\begin{equation}
    (1 - \matr{X}_{ij}) = 1 \iff \compal(j) \in \prefix{G_i} \text{ and }
    \matr{X}^T_{ij} = 1 \iff \compal(j) \in \prefixopen{G_i} \label{eq:X_prefix}
\end{equation}
We can now rewrite the Eq.~\ref{eq:rewrite_const_vio} to match the property in Eq.~\ref{eq:violation}, with abbreviations $C^\clm_{\idr_r}((e,t_\mode)) = \F(p_r, t)((\varepsilon, \mode^{-1}(\idr_r)))$ and $C^\clm_{\idr_r}((e,t_\mode)) = \F(t,p_r)((\varepsilon, \mode^{-1}(\idr_r)))$. Note that $\matr{C}^\clm_{jk} = C^\clm_{\IdR(k)}(\compal(j))$ and $\matr{C}^\rls_{jk} = C^\rls_{\IdR(k)}(\compal(j))$ for every $j \in [1..n_\compal]$ and $k \in [1..n_r]$. For every $i \in [1..n_\compal]$ and $k \in [1..n_r]$ with $\idr = \IdR(k)$, the following holds:
\begin{align}
    && (1 - \matr{X}_{i\bullet}) \matr{C}^\clm_{\bullet k} - \matr{X}^T_{i\bullet} \matr{C}^\rls_{\bullet k} &\leq |\idr| &\\
    &\iff & \sum_{j \in [1..n_\gamma]} \left( (1 - \matr{X}_{ij}) \matr{C}^\clm_{jk} - \matr{X}^T_{ij} \matr{C}^\rls_{jk} \right) &\leq |\idr| &\\
    &\iff & \sum_{g \in \prefix{G_i}} C^\clm_\idr(g) - \sum_{g \in \prefixopen{G}} C^\rls_\idr(g) &\leq |\idr| & \text{(By Eq.~\ref{eq:X_prefix})} \\
    &\iff & \sum_{g \in \prefixopen{G_i}} \left( C^\clm_\idr(g) - C^\rls_\idr(g) \right) + \sum_{g \in G} C^\clm_\idr(g) &\leq |\idr| & \\
    &\iff & |\idr| - \sum_{g \in \prefixopen{G_i}} \left( C^\clm_\idr(g) - C^\rls_\idr(g) \right) + \sum_{g \in G} C^\clm_\idr(g) & \geq 0 & \\
    &\iff & \cutmarking(\prefixopen{G})(p_r)((\varepsilon,\idr_r)) + \sum_{g \in G} C^\clm_\idr(g) & \geq 0 & \text{(By Def. \ref{def:cutmarking})}
\end{align}
\qed
\end{proof}

Next, we go over two important properties of the ILP problem. In Lemma~\ref{lem:ilp_solution_exists}, we show that a solution respecting the constraints always exists for any composed alignment $\compal$, and in Lemma~\ref{lem:cost0} we show that if and only if $\compal$ is not violating as defined in Def.~\ref{def:violations}, the cost of the solution is 0.
\begin{lemma}{(There always exists a solution to the ILP problem.)}\label{lem:ilp_solution_exists}
With a composed alignment $\compal = \doublecup_{\idc \in \Idc} \gamma_c$ and the ILP problem formulated as above, there exists a solution for $\matr{X}$ such that all constraints hold.
\end{lemma}
\begin{proof}
We show by construction that there is always a solution $\matr{X}'$ to the ILP problem, which respects all constraints. Let there be a (possibly arbitrary) order in $\Idc$, such that $c_x$ denotes the $x^{\text{th}}$ case identifiers for every $x \in [1..n_c]$, with $n_c = |\Idc|$.
\begin{equation}
    \matr{X}' =
    \begin{bmatrix}
    \matr{R}_{c_1} & \matr{1} & \cdots & \matr{1}  \\
    \matr{0} & \matr{R}_{c_2} & \cdots & \matr{1}  \\
    \vdots & \vdots & \ddots & \vdots  \\
    \matr{0} & \matr{0} & \cdots & \matr{R}_{c_{n_c}}
    \end{bmatrix}
\end{equation}
with $\matr{R}_c = \matr{R}_{I_c I_c}$ for each $c \in \Idc$ the $|\gamma_c|\times|\gamma_c|$ submatrix containing only the elements of the case.
\begin{enumerate}
    \item[Const.~\ref{eq:const_c}] $\matr{X}'$ trivially respects this constraint as it contains $\matr{R}_c$ for each $c \in \Idc$;
    \item[Const.~\ref{eq:const_rev_rem}] Let $i \in I_{c_x}$ and $j \in I_{c_y}$ be two indices. If $x = y$, this trivially holds, since $\matr{X}'$ contains $\matr{R}_{c_x}$. Otherwise, by construction of $\matr{X}_{ij}$, $\matr{R}_{ij} = 1 \wedge \matr{X}_{ij} = 0 \implies j < i \implies \matr{X}_{ji} = 1$;
    \item[Const.~\ref{eq:const_trans_clos}] Let $i,j,k \in I_{c_x},I_{c_y},I_{c_z}$ respectively. When $x=y=z$, this trivially holds, since $\matr{X}'$ contains $\matr{R}_{c_x}$ which is transitively closed. Otherwise, we know by construction of $\matr{X}'$ that with $x \neq y$, $\matr{X}'_{ij} = 1 \iff x < y$ and $\matr{X}'_{ij} = 0 \iff x > y$ which we use two prove the two cases:
    \begin{enumerate}
        \item[(1)] $\matr{X}'_{ij} = 1 \wedge x \neq y \implies x < y$. $\matr{X}'_{jk} = 1 \implies y \leq z \implies x < z \implies \matr{X}'_{ik} = 1$;
        \item[(2)] $\matr{X}'_{jk} = 1 \wedge y \neq z \implies y < z$. $\matr{X}'_{ij} = 1 \implies x \leq y \implies x < z \implies \matr{X}'_{ik} = 1$.
    \end{enumerate}
    Hence $\matr{X}'$ is transitively closed;
    \item[Const.~\ref{eq:const_vio}] For every $x \in [1..n_c]$ and every $i \in I_{c_x}$, let $J_i = \set{ j \mid j \in [1..n_\compal], 1 - \matr{X}'_{ij} = 1 }$ and $J_i' = \set{ j \mid j \in [1..n_\compal], \matr{X}'^T_{ij} = 1 }$, note that $J' \subseteq J$.
    We know by construction of $\matr{X}'$ that for every $y \in [1..n_c]$ and every $j \in I_{c_y}$, $1 - \matr{X}'_{ij} = 1 \implies y \leq x$ and $\matr{X}'^T_{ij} = 1 \implies y \leq x$. Therefore
    \begin{equation}
        \forall_{y \in [1..n_c] \setminus \set{x}} I_{c_y} \subseteq J' \vee I_{c_y} \cap J' = \emptyset
    \end{equation}
    holds. Since $\gamma_{c_y}$ is an alignment for every $y \in [1..n_c]$, we have for every resource index $k \in [1..n_r]$, $\sum_{j \in I_{c_y}}(\matr{C}^\clm_{jk} - \matr{C}^\rls_{jk}) = 0$ which together imply that
    \begin{equation}
    \begin{gathered}
        \sum_{j \in J} \matr{C}^\clm_{jk} - \sum_{j \in J'} \matr{C}^\rls_{jk} =
        \sum_{y \in [1..n_c]} \left(\sum_{j \in J \cap I_{c_y}} \matr{C}^\clm_{jk} - \sum_{j \in J' \cap I_{c_y}} \matr{C}^\rls_{jk} \right) \\
        = \sum_{j \in J \cap I_{c_x}} \matr{C}^\clm_{jk} - \sum_{j \in J' \cap I_{c_x}} \matr{C}^\rls_{jk} \leq |\IdR(k)|
    \end{gathered}
    \end{equation}
    since $\gamma_{c_x}$ is a non-violating alignment.
\end{enumerate}
\qed
\end{proof}

\begin{lemma}{(Cost=0 if and only if the composed alignment is not violating.)}\label{lem:cost0}
With a composed alignment $\compal = \doublecup_{\idc \in \Idc} \gamma_c$, the solution to the ILP problem formulated above has zero cost if and only if $\compal$ is not violating as defined in Def.~\ref{def:violations}, \ie Eq.~\ref{eq:violation} does not hold for $\compal$.
\end{lemma}
\begin{proof}
We prove the two sides of the bi-implication:
\begin{itemize}
    \item[$(\implies)$] Assuming the cost is zero, we have for every $i,j \in [1..n_\compal]$, $\matr{R}_{ij} = 1 \implies \matr{X}_{ij} = 1$ by the objective function, and therefore $\set{ (\compal(i), \compal(j)) \mid i,j \in [1..n_\compal],\matr{R}_{ij} = 1} = \prec_\matr{R} \subseteq \prec_\matr{X} = \set{ (\compal(i), \compal(j)) \mid i,j \in [1..n_\compal],\matr{X}_{ij} = 1}$. Since $\matr{X}$ respects Const.~\ref{eq:const_vio}, $\matr{R}$ is not violating either;
    \item[$(\impliedby)$] Assuming $\compal$ is not violating, there exists a $\compal' \in \mathcal{S}(\compal)$, such that the resource capacities are respected. Any $\matr{X}$ such that $\prec_\compal \subseteq \prec_{\compal'} \subseteq \prec_\matr{X}$ results in zero cost, since for every $i,j \in [1..n_\compal]$, $\matr{R}_{ij} = 1 \implies \matr{X}_{ij} = 1$. Furthermore, the cost can not be negative as $(1 - \matr{R}_{ij}) \matr{X}_{ij} \geq 1$ for every $i,j\in [1..n_\compal]$.
\end{itemize}
\qed
\end{proof}

Lastly, to fulfill the assumption in Sec.~\ref{sec:approximation} that the intervals around violating antichains are alignable, we show in Lemma~\ref{lem:interval_alignable} that this is the case for each interval in $W(\compal'_V)$.
\begin{lemma}{(Each interval $\interval{A}{B} \in W(\compal'_V)$ is alignable)}\label{lem:interval_alignable}
Let $\compal = \doublecup_{\idc \in \Idc} \gamma_c$ be a composed alignment and $W(\compal'_V)$ with $\compal'_V = \cup_{i,j\in[1..n_\compal], \matr{X}_{ij} - \matr{R}_{ij} = 1 } \interval{\compal(j)}{\compal(i)}$ the set of intervals obtained from the ILP problem formulated above with solution $\matr{X}$. For every $\interval{A}{B} \in W(\compal'_V)$, $\interval{A}{B}$ is alignable, \ie (by Def.~\ref{def:align_interval}) $m_A \xrightarrow{*} m_B$ with $m_A = \cutmarking(\prefixopen{A})$ and $m_B = \cutmarking(\prefix{B})$.
\end{lemma}
\begin{proof}
Let $\interval{A}{B} \in W(\compal'_V)$ be any interval in $W(\compal'_V)$. We prove that $m_A \xrightarrow{*} m_B$ by construction of a subalignment $\gamma_{AB} = (\bar\gamma_{AB}, \prec_{\gamma_{AB}})$ constructed as follows:
\begin{align}
    \bar\gamma_{AB} &= \bigcup_{(e, t_\mode) \in \interval{A}{B} \cap \Gamma_s} \sset{ (\gg, t_\mode), (e, \gg) } \\ 
    \prec_{\gamma_{AB}} &= ( \sset{ \spair{(\gg, t_\mode(i))}{(\gg, t_\mode(j))} \mid i,j \in [1..n_\gamma], \gamma(i), \gamma(j) \in \interval{A}{B} \setminus \Gamma_l, \matr{X}_{ij} = 1 } \\
     &~\cup \set{((e, t_\mode), (e', t'_\mode)) \mid (e, t_\mode), (e', t'_\mode)) \in \prec_{\gamma'_L}, (e, \gg),(e',\gg) \in \gamma_{AB}} )^+\nonumber
\end{align}
\ie all synchronous moves are split into a corresponding model and log move and the partial order for model moves is defined by $\matr{X}$ and for the log moves by the event log's partial order $\prec_L$.
There is a $\sigma \in N(\gamma_{AB})^*$ such that $m_A \xrightarrow{\sigma} m_B$ because of two properties of $\gamma_{AB}$:
\begin{itemize}
    \item By Const.~\ref{eq:const_vio} and Lemma~\ref{lem:const_vio}, $\forall_{G \in \mathcal{A}^+(\gamma_{AB}), \idr \in supp(\IdR)} \neg\viol(G)$;
    \item For every $c \in \Idc$, we have, by construction of $\gamma_{AB}$ and the ILP's constraints, $\gamma_{AB}\proj_{T_c} = \interval{A}{B}\proj_{T_c}$, with $T_c = \set{t_\mode \mid t_\mode \in T_\mode, \exists_{v_c \in Var_c, v_r \in Var_r)} \mode((v_c,v_r))=(c,*) }$. Therefore, $\gamma_{AB}\proj_c = \gamma_{AB} \cap \Gamma_c$ is an alignment, with $\Gamma_c = \set{(e,t_\mode) \mid (e,t_\mode) \in \gamma_{AB}, \ecase(e) = c \vee t_\mode \in T_c}$.
\end{itemize}
\qed
\end{proof}
\fi

\end{document}